\documentclass[12pt]{article}
 \pdfoutput=1
\usepackage[dvipdfm]{graphicx}
\usepackage{amsthm}
\usepackage[usenames]{color}
\usepackage{xspace,colortbl}
\usepackage{latexsym,url,amscd,amssymb}
\usepackage{amsmath}
\usepackage{graphics}
\usepackage{epstopdf}
\usepackage{hyperref}
\usepackage{subcaption}
\usepackage{booktabs} %调整表格线与上下内容的间隔
\usepackage{multirow}
\usepackage{cases}
\usepackage{caption}
\usepackage{algorithm}
\usepackage{algpseudocode}
\usepackage{listings}
\usepackage{rotating}
\usepackage{framed}
\usepackage{natbib}
\usepackage{url} % not crucial - just used below for the URL

\DeclareMathOperator*{\argmin}{arg\,min}

\DeclareMathOperator*{\sign}{sign}
\def\prox{{\rm Prox}}
\def\A{{\mathcal A}}

\def\I{{\mathcal I}}

\newtheorem{lemma}{Lemma}[section]

\newtheorem{theorem}{Theorem}[section]

\newtheorem{remark}{Remark}[section]

\begin{document}

\def\spacingset#1{\renewcommand{\baselinestretch}%
{#1}\small\normalsize} \spacingset{1}

%%%%%%%%%%%%%%%%%%%%%%%%%%%%%%%%%%%%%%%%%%%%%%%%%%%%%%%%%%%%%%%%%%%%%%%%%%%%%%

%\if0\blind
{
\title{\bf A Global Two-stage Algorithm for Non-convex Penalized High-dimensional Linear Regression Problems}

\author{
Peili Li\thanks{School of Statistics, KLATASDS-MOE, East China Normal University, Shanghai 200062, PR China (Email: plli@sfs.ecnu.edu.cn).}
\and
Min Liu\thanks{School of Mathematics and Statistics, Wuhan University, Wuhan 430072, PR China (Email: mliuf@whu.edu.cn).}
\and
Zhou Yu\thanks{School of Statistics, KLATASDS-MOE, East China Normal University, Shanghai 200062, PR China (Email: zyu@stat.ecnu.edu.cn).}
}

%\author{Meixia Yang\\
%School of Mathematics and Statistics, Henan University, Kaifeng, China\\
%\and
%Yunhai Xiao\thanks{The work of Y. Xiao is supported by the National Natural Science Foundation of China (Grant No. 11971149).}\hspace{.2cm}\\
%Institute of Applied Mathematics, Henan University, Kaifeng, China\\
%\and
%Peili Li\\
%School of Statistics, KLATASDS-MOE, ECNU, Shanghai, China\\
%\and
%Hanbing Zhu\\
%School of Statistics, KLATASDS-MOE, ECNU, Shanghai, China
%}
\date{}
\maketitle
}
%\fi

%\if1\blind
%{
%\bigskip
%\bigskip
%\bigskip
%\begin{center}
%{\LARGE\bf Title}
%\end{center}
%\medskip
%} \fi

\bigskip
\begin{abstract}
%Semiparametric regression is a fusion between traditional parametric regression analysis and nonparametric regression methods.
By the asymptotic oracle property, non-convex penalties represented by minimax concave penalty (MCP) and smoothly clipped absolute deviation (SCAD) have attracted much attentions in high-dimensional data analysis, and have been widely used in signal processing, image restoration, matrix estimation, etc.
However, in view of their non-convex and non-smooth characteristics, they are computationally challenging. Almost all existing algorithms converge locally, and the proper selection of initial values is crucial. Therefore, in actual operation, they often combine a warm-starting technique to meet the rigid requirement that the initial value must be sufficiently close to the optimal solution of the corresponding problem. In this paper, based on the DC (difference of convex functions) property of MCP and SCAD penalties, we aim to design a global two-stage algorithm for the high-dimensional least squares linear regression problems. A key idea for making the proposed algorithm to be efficient is to use the primal dual active set with continuation (PDASC) method, which is equivalent to the semi-smooth Newton (SSN) method, to solve the corresponding sub-problems. Theoretically, we not only prove the global convergence of the proposed algorithm, but also verify that the generated iterative sequence converges to a d-stationary point. In terms of computational performance, the abundant research of simulation and real data show that the algorithm in this paper is superior to the latest SSN method and the classic coordinate descent (CD) algorithm for solving non-convex penalized high-dimensional linear regression problems.
\end{abstract}

\noindent%
\vfill

%\newpage
\spacingset{1.75} % DON'T change the spacing!
\setcounter{equation}{0}
\section{Introduction}\label{intro}

In this paper, we mainly consider the following high-dimensional linear regression model:
\begin{align}\label{LR}
y=X \beta^{*}+\varepsilon,
\end{align}
where $y\in \mathbb{R}^{n}$ is the response vector, $X\in \mathbb{R}^{n\times p}$ is the design matrix, $\varepsilon\in \mathbb{R}^{n}$ is the noise vector, and $\beta^{*}$ is the underlying regression coefficient. In the high-dimensional settings, the number of predictors $p$ is usually larger or much larger than the number of observations $n$. At this time, we usually assume that $\beta^{*}$ is sparse, that is, only a small part of its elements are non-zero. If this idea is expressed in the parameter estimation models, it is natural to add the constraint $\|\beta\|_0\leq s$, where $\|\beta\|_0$ denotes the number of non-zero elements in $\beta$, and $s>0$ is a tuning parameter which controls the sparsity level. However, the non-convexity and discontinuity of the $\ell_0$ pseudo-norm make it NP-hard to solve the corresponding problems \cite{N1995}. Especially in the high-dimensional settings, it is very challenging to design a feasible algorithm that can achieve accurate solutions. Therefore, various surrogates of the $\ell_0$ pseudo-norm have been proposed in the existing literature and have been widely studied in statistics, optimization, computational mathematics, machine learning and other fields.

The first type of surrogate functions is mainly the well-known $\ell_1$ norm \cite{CDS2001,FJL2014,T1996}, and its corresponding Lagrangian form of least squares  linear regression model is the following convex but non-smooth minimization problem:
\begin{equation}\label{L1}
\min_{\beta\in\mathbb{R}^p}\Big\{\frac{1}{2}\|X \beta-y\|^2+\lambda \|\beta\|_1\Big\},
\end{equation}
where $\|\beta\|_1=\sum_{i=1}^{p}|\beta_i|$ denotes the $\ell_1$ norm of the vector $\beta$, $\lambda>0$ is a regularization parameter. In view of the good characteristics of the above model, it has received extensive attention in different application fields. Theoretically, under certain conditions on the design matrix $X$ and the sparsity level of the underlying regression coefficient $\beta^{*}$, the minimizers of (\ref{L1}) have attractive statistical properties \cite{CT2005,MB2006,ZY2006}. Numerically, the convexity of (\ref{L1}) has led to many fast and effective algorithms, such as least angle regression (LARS)\cite{EHJT2004}, alternating direction method of multipliers (ADMM) \cite{BPC2011}, coordinate descent(CD) method  \cite{WL2008} and semi-smooth Newton (SSN) method (or equivalent primal dual active set (PDAS) algorithm) \cite{HIK2002,LST2018} etc. It is worth emphasizing that the PDAS algorithm in \cite{FJL2014} not only has the local superlinear convergence which can be obtained by reformulating it in the SSN framework, but also has the locally one step convergence under certain conditions. In addition, the continuation technique on the regularization parameter globalizes the convergence of the algorithm. In this paper, we will apply it to solve internal sub-problems, and one can see Section \ref{BSSN} for details.

Although the convexity of $\ell_1$ penalty makes the corresponding problem computationally attractive, there still exists bias in its estimator. Therefore, scholars proposed the second type of surrogate functions for $\ell_0$ pseudo-norm, which mainly contains some non-convex penalties, such as the minimax concave penalty (MCP) \cite{Z2010}, the smoothly clipped absolute deviation (SCAD) penalty \cite{FL2001}, capped $\ell_1$ \cite{Z2010b} and bridge \cite{FF1993,F1998} etc. Numerous studies have shown that, compared with a convex relaxation with the $\ell_1$ norm, a proper non-convex penalty method can achieve a sparse estimation with fewer measurements, and is more robust against noises \cite{C2007,CG2014}. Therefore, non-convex penalties have been widely used in various sparse learning problems \cite{BH2011,C2007,CG2014,GZL2013,HJJ2019,LYG2017,MFH2011}.

In this paper, we mainly focus on the least squares regression model with MCP or SCAD penalty, i.e.,
\begin{equation}\label{primal}
\min_{\beta\in\mathbb{R}^p}\Big\{\frac{1}{2}\|X \beta-y\|^2+\sum_{i=1}^{p}\rho(\beta_i;\lambda,\tau)\Big\},
\end{equation}
where $\rho(\cdot;\lambda,\tau)$ is the MCP or SCAD penalty, which are respectively defined by
\begin{align}\label{mcp}
\rho_{mcp}(t;\lambda,\tau):=\lambda \int_{0}^{|t|}\max\Big(0,1-\frac{s}{\lambda\tau}\Big)ds=\left\{
  \begin{array}{ll}
    \frac{\lambda^2\tau}{2}, &|t|> \lambda \tau,\\
    \lambda (|t|-\frac{t^2}{2\lambda \tau}), & |t|\leq\lambda \tau, \quad \tau>1,
  \end{array}\right.
\end{align}
\begin{align}\label{scad}
\rho_{scad}(t;\lambda,\tau):=\lambda \int_{0}^{|t|}\min\Big(1,\frac{\max(0,\lambda \tau-s)}{\lambda(\tau-1)}\Big)ds
=\left\{
  \begin{array}{lll}
    \frac{\lambda^2(\tau+1)}{2}, & |t|>\lambda \tau,\\
    \frac{\lambda\tau|t|-\frac{1}{2}(t^2+\lambda^2)}{\tau-1}, & \lambda<|t|\leq \lambda \tau, \quad \tau>2.\\
    \lambda |t|, &|t|\leq \lambda,
  \end{array}\right.
\end{align}
Here $\tau$ is a given parameter which controls the concavity of the corresponding penalty. When proposing MCP and SCAD penalties, their authors established that the regression models with MCP and SCAD penalties have the so-called oracle property, that is, in an asymptotic sense, they perform as well as if the analyst had known in advance which coefficients were zero and which were nonzero.

However, non-convex and non-smooth characteristics of the objective function make the numerical calculation of model (\ref{primal}) very challenging. There are several typical algorithms in the existing literature, and here is a simple summary in chronological order. Firstly, the authors of \cite{FL2001,HL2005} proposed a local quadratic approximation (LQA) algorithm and its slightly perturbed version. They suggested iteratively, locally approximating the penalty function by a quadratic function, and then using a modified Newton-Raphson algorithm to solve the corresponding problem. However, the behavior of deleting small coefficients or choosing the size of perturbation will cause numerical instability. To overcome this difficulty, Zou and Li \cite{ZL2008} proposed a new unified algorithm based on the local linear approximation (LLA), and calculated the resulting LASSO problem by LARS algorithm. However, LLA used the path-tracing LARS algorithm to update the regression coefficients, so it is inherently inefficient to some extent. Then, the coordinate descent (CD) type algorithms were designed for the least squares regression models penalized by MCP and SCAD \cite{BH2011,MFH2011}. The numerical results showed that the performance of this algorithm is better than that of LLA. However, the CD-type algorithm requires many iterations in the pursuit of high accuracy, because its convergence rate is sub-linear or linear locally \cite{LP2017}. In addition, it has been proved that each non-convex surrogate function of $\ell_0$ pseudo-norm can be expressed as the difference of two convex functions \cite{APX2017,LPLV2015}. Therefore, based on the DC (difference of convex functions) property of the non-convex functions, Li et al. \cite{LYG2017} proposed a DC proximal Newton (DCPN) method for the general nonlinear problems with non-convex penalty. They firstly used multistage convex relaxation to transform the original optimization into sequences of LASSO regularized nonlinear regressions. Then, in each stage, they used the second order Taylor expansion to approximate the nonlinear loss functions, and adopted the Proximity Newton method in \cite{LSS2014} to solve the convex sub-problem. Under the conditions of locally restricted strong convexity and Hessian smoothness, they proved their algorithm is locally quadratic convergent within each stage of convex relaxation. Recently, Shi et al. \cite{SHJ2018} and Huang et al. \cite{HJJ2019} respectively proposed SSN and PDAS algorithms for the model (\ref{primal}), and their convergence rates are all locally super-linear.

After in-depth study of the relevant literature, we can find that above-mentioned algorithms are all locally convergent, so they generally combine various warm-starting techniques in actual operations. This inspires us to design an effective calculation method with global convergence to weaken the rigid requirement that the initial value must be sufficiently close to the optimal solution. Here we will design a global two-stage algorithm based on the DC expression of MCP and SCAD penalties. From \cite{APX2017,LPLV2015,TWST2020}, we know that MCP and SCAD penalties can be reformulated as:
\begin{align}
\sum_{i=1}^{p}\rho_{mcp}(\beta_i;\lambda,\tau)&=\lambda \|\beta\|_1-q_{mcp}(\beta), \quad \tau>1,\label{mcpreform}\\ \sum_{i=1}^{p}\rho_{scad}(\beta_i;\lambda,\tau)&=\lambda \|\beta\|_1-q_{scad}(\beta),\quad \tau>2,\label{scadreform}
\end{align}
where $q_{mcp}(\beta)=\sum_{i=1}^{p}q_{mcp}(\beta_i;\lambda,\tau)$, $q_{scad}(\beta)=\sum_{i=1}^{p}q_{scad}(\beta_i;\lambda,\tau)$, and
$$q_{mcp}(t;\lambda,\tau)=\left\{
  \begin{array}{ll}
    \lambda |t|-\frac{\lambda^2\tau}{2}, &|t|> \lambda \tau\\
   \frac{t^2}{2\tau}, & |t|\leq\lambda \tau
  \end{array}\right.,\quad q_{scad}(t;\lambda,\tau)=\left\{
  \begin{array}{lll}
    \lambda |t|-\frac{\lambda^2(\tau+1)}{2}, & |t|>\lambda \tau\\
    \frac{(|t|-\lambda)^2}{2(\tau-1)}, & \lambda<|t|\leq \lambda \tau\\
    0, &|t|\leq \lambda
  \end{array}\right..$$ The functions $q_{mcp}(\beta)$ and $q_{scad}(\beta)$ are continuously differentiable with
$$\frac{\partial q_{mcp}(\beta)}{\partial \beta_i}=\left\{
  \begin{array}{ll}
    \lambda \sign(\beta_i), &|\beta_i|> \lambda \tau\\
   \frac{\beta_i}{\tau}, & |\beta_i|\leq\lambda \tau
  \end{array}\right., \quad \frac{\partial q_{scad}(\beta)}{\partial \beta_i}=\left\{
  \begin{array}{lll}
    \lambda \sign(\beta_i), & |\beta_i|>\lambda \tau\\
    \frac{\sign(\beta_i)(|\beta_i|-\lambda)}{\tau-1}, & \lambda<|\beta_i|\leq \lambda \tau\\
    0, &|\beta_i|\leq \lambda
  \end{array}\right..$$
Therefore, the original model (\ref{primal}) can be rewritten as follows,
\begin{equation}\label{dcmodel}
\min_{\beta\in\mathbb{R}^p}\Big\{f(\beta):=\frac{1}{2}\|X \beta-y\|^2+\lambda \|\beta\|_1-q(\beta)\Big\},
\end{equation}
where $q: \mathbb{R}^{p}\rightarrow \mathbb{R}$ is $q_{mcp}$ or $q_{scad}$, which is a convex smooth function. Together with the motivation from global and super-linear proximal majorization-minimization (PMM) algorithm in \cite{TWST2020}, which is proposed for nonconvex square-root-loss regression problems, we are thus inspired to adopt the PMM framework for solving the least squares model (\ref{dcmodel}). A key idea for making the proposed algorithm to be efficient is to use the PDASC algorithm for solving the corresponding sub-problems. Specifically, in the first stage, by directly removing the second term $-q(\beta)$ and adding a proximal term $\frac{\sigma}{2}\|\beta\|^2$, we will use the PDASC method in \cite{FJL2014} to solve the obtained convex sub-problem, which can get an initial point of the second stage. Then in the second stage, we linearize the second term $-q(\beta)$ with respect to the current iteration $\beta^k$ and add an appropriate proximal term $\frac{\sigma}{2}\|\beta-\beta^k\|^2$, then directly use the PDASC method to iteratively solve the resulting problem.

The remainder of this paper is organized as follows. In Section \ref{pre}, we present some preliminaries for our subsequent developments. In Section \ref{alg}, we describe the two-stage algorithm detaily. We establish the algorithm's convergence in Section \ref{the}. In Section \ref{num}, we report numerical experiments to show the efficiency of the algorithm, and do performance comparisons with the latest SSN method and the classic CD algorithm. Finally, we conclude our paper in Section \ref{con}.

\section{Preliminaries}\label{pre}
\setcounter{equation}{0}
%%%%%%%%%%%%%%%%%%%%%%%%%%%%%%%%%%%%%%%%%%%%%%%%%%%%%%%%%%%%%%%%%%%%%%%%%%%%%%%%%%%%%%%%%%%%%%%%%%%%%%%%%%%%%%%%%%%%%%%%%

We denote the set of all proper lower semicontinuous convex functions on $\mathbb{R}^p$ as $\mathcal{L}(\mathbb{R}^p)$. For a given $f\in \mathcal{L}(\mathbb{R}^p)$, The proximal mapping of $f$ is defined as
$$Prox_{f}(x):=\argmin_{y\in \mathbb{R}^p}\Big\{ f(y)+\frac{1}{2}\|y-x\|^2\Big\}, \quad \forall x\in \mathbb{R}^p.$$
Then, from \cite{MSX2011}, we have
\begin{equation}\label{pri1}
z\in \partial f(y) \Leftrightarrow y=Prox_{f}(y+z).
\end{equation}
The proximal operator of $\|\cdot\|_1$ is given by the pointwise soft-thresholding
operator \cite{DJ1995}:
\begin{equation}\label{pri2}
Prox_{\lambda \|\cdot\|_1}(x)=S_{\lambda}(x),
\end{equation}
where
\begin{equation}\label{soth}
y=S_{\lambda}(x) \Leftrightarrow y_i=\max\{|x_i|-\lambda,0\}\text{sign}(x_i).
\end{equation}

The subdifferential of any $f\in \mathcal{L}(\mathbb{R}^p)$ is a set-value mapping defined by
$$\partial f(x):=\{z\in \mathbb{R}^p:f(y)\geq f(x)+\langle z, y-x\rangle, \forall y\in \mathbb{R}^p\}.$$
The subdifferential of $f=\|x\|_1$ is the pointwise set-value sign function $\text{Sign(x)}$ \cite{DJ1995}, i.e.,
\begin{equation}\label{sub1}
z\in \text{Sign(x)} \Leftrightarrow z_i\left\{
\begin{array}{lll}
=1, &x_i>0\\
\in[-1,1], &x_i=0\\
=-1, &x_i<0
\end{array}
\right..
\end{equation}
The classical Fermat's rule for proper lower semicontinuous convex functions \cite{R1996} asserts
\begin{equation}\label{cfr}
\textbf{0}\in \partial f(x^*)\Leftrightarrow x^* is \ a \ global \ minimizer \ of \ f,
\end{equation}
where $\textbf{0}$ denotes a column vector whose elements are all 0. If the function $f$ is locally Lipschitz continuous near $x^*$ and directionally differentiable at $x^*$, then $0\in \partial f(x^*)$ is equivalent to the directional-stationarity (d-stationarity) of $x^*$, that is
$$f'(x^*;h):=lim_{\delta\rightarrow 0}\frac{f(x^*+\delta h)-f(x^*)}{\delta}\geq0, \forall h \in \mathbb{R}^p.$$
In this paper, we will prove that the iterative sequence of the proposed algorithm converges to a d-stationarity point of problem (\ref{dcmodel}).

%%%%%%%%%%%%%%%%%%%%%%%%%%%%%%%%%%%%%%%%%%%%%%%%%%%%%%%%%%%%%%%%%%%%%%%%%%%%%%%%%%%%%%%%%%%%%%%%%%%%%%%%%%
%%%%%%%%%%%%%%%%%%%%%%%%%%%%%%%%%%%%%%%%%%%%%%%%%%%%%%%%%%%%%%%%%%%%%%%%%%%%%%%%%%%%%%%%%%%%%%%%%%%%%%%%%%
\section{Algorithm}\label{alg}
\setcounter{equation}{0}
In this section, we will propose a two-stage proximal majorization-minimization (PMM) algorithm for model (\ref{primal}), and the internal sub-problem with $\ell_1$ penalty will be approximately solved by the primal dual active set with continuation (PDASC) method in \cite{FJL2014}.

\subsection{PMM algorithm}
The PMM algorithm contains two stages, where the first stage provides a good initial point for the second stage. Another key idea to make PMM algorithm effective is to use the PDASC algorithm for solving the corresponding subproblems. Specifically, in the first stage, we get a nonsmooth convex subproblem with $\ell_1$ penalty by directly removing the concave term $-q(\beta)$ and adding a proximal term $\frac{\sigma}{2}\|\beta\|^2$. Then we use PDASC method to approximately solve the obtained subproblem so that the corresponding KKT residual satisfies a prescribed termination criterion. Next, the solution obtained in the first stage is used as the initial value of the second stage. In the second stage, we linearize the concave term $-q(\beta)$ with respect to the current iteration $\beta^k$ and add an appropriate proximal term $\frac{\sigma}{2}\|\beta-\beta^k\|^2$. Then we also use PDASC to solve the corresponding convex sub-problem so that the error vector satisfies a preset accuracy condition. In this way, the second stage is looped and the penalty parameter $\sigma$ is updated iteratively until the iteration sequence satisfies the termination condition given in advance.

Given $\sigma>0$, $\tilde{\beta}\in \mathbb{R}^p$ and $\tilde{v}\in \mathbb{R}^p$, we consider the following minimization problem in each iteration:
\begin{equation}\label{linemode2}
\min_{\beta\in\mathbb{R}^p}J(\beta;\sigma,\tilde{\beta},\tilde{v}):=\frac{1}{2}\|X \beta-y\|^2+\lambda \|\beta\|_1-q(\tilde{\beta})-\langle \tilde{v},\beta-\tilde{\beta} \rangle+\frac{\sigma}{2}\|\beta-\tilde{\beta}\|^2.
\end{equation}
Obviously, the above model is a convex problem with $\ell_1$ penalty, which can be effectively solved by PDASC method. Next, we summarize the iterative framework of PMM algorithm in Algorithm \ref{algorithm1}.

\begin{framed}\label{algorithm1}
\noindent
{\bf PMM algorithm}
\vskip 1.0mm \hrule \vskip 1mm
\noindent
\textbf{Step 1.} Take $\sigma^1>0$, $\sigma^{2,0}>0$. Compute
\begin{align}\label{step1}
\beta^0 = \argmin_{\beta\in \mathbb{R}^p}\{J(\beta;\sigma^1,\textbf{0},\textbf{0})\}
\end{align}
by PDASC method such that the corresponding KKT residual satisfies a prescribed termination criterion.
For $k=0,1,2,\ldots$, do the following operations iteratively.\\
\textbf{Step 2.} Compute
\begin{align*}
\beta^{k+1}=\argmin_{\beta\in \mathbb{R}^p}\{J(\beta;\sigma^{2,k},\beta^k,\nabla q(\beta^k))+\langle \delta^k, \beta-\beta^k\rangle\}
\end{align*}
by PDASC method such that the error vector $\delta^k$ satisfies
\begin{align}\label{err1}
\|\delta^k\|\leq \frac{\sigma^{2,k}}{4}\|\beta^{k+1}-\beta^{k}\|.
\end{align}
\textbf{Step 3.} Check the prescribed stopping condition, if stop, denote the last iteration by $\hat{\beta}$. Else, update $\sigma^{2, k+1}=\gamma \sigma^{2,k}$ with $\gamma\in (0,1)$ and set $k:=k+1$.
%\\[2mm]
%\textbf{Step 4.} Set $k:=k+1$.
\end{framed}
\begin{remark}
It should be pointed out that, we do not need to calculate the dual problem of the corresponding subproblem. Because the sub-problem here is essentially a convex problem with $\ell_1$ penalty, which can be directly and effectively solved by the PDASC method. This part is different from \cite{TWST2020}.
\end{remark}
\begin{remark}
Through the verification of many experiments and the communication with the authors in \cite{TWST2020}, we found that if we use PDASC to solve $\min\limits_{\beta\in \mathbb{R}^p}\{J(\beta;\sigma^{2,k},\beta^k,\nabla q(\beta^k))\}$ in the second stage so that the corresponding KKT residual satisfies a prescribed  accuracy, such as $1e-6$, then the condition (\ref{err1}) is automatically contented. Therefore, in our subsequent numerical experiments, the termination conditions of all sub-problems are set as the corresponding KKT residuals are sufficiently small. And the inequality (\ref{err1}) is mainly used for theoretical analysis.
\end{remark}

\subsection{The PDASC method for sub-problems}\label{BSSN}
From \cite{FJL2014}, we can see that the design idea of PDASC method is inspired by the first order optimality system of (\ref{linemode2}), which can be seen in the following Lemma \ref{lemma1}.

\begin{lemma}\label{lemma1}
$\beta^*\in \mathbb{R}^{p}$ is a global minimizer of (\ref{linemode2}) if and only if there exists a $d^*\in \mathbb{R}^p$ such that the following KKT system holds:
\begin{align}\label{KKT2}
(X^{\top}X&+\sigma I)\beta^*+d^*=X^{\top}y+\tilde{v}+\sigma \tilde{\beta},\\
&\beta^*=S_{\lambda}\Big(\beta^*+d^*\Big).
\end{align}
\end{lemma}
\begin{proof}
By (\ref{cfr}), we can have
\begin{align*}
\beta^*\in \mathbb{R}^p \ is \ a \ minimizer \ of \ (\ref{linemode2}) \Leftrightarrow \textbf{0} \in \partial J(\beta^*;\sigma,\tilde{\beta},\tilde{v}).
\end{align*}
Obviously, $\partial J(\beta^*;\sigma,\tilde{\beta},\tilde{v})=(X^{\top}X+\sigma I)\beta^*-X^{\top}y-\tilde{v}-\sigma \tilde{\beta}+\lambda \partial \|\cdot\|_1(\beta^*)$.
Therefore, there exists $d^* \in \lambda \partial \|\cdot\|_1(\beta^*)$ such that
\begin{align}\label{ome1}
(X^{\top}X+\sigma I)\beta^*+d^*=X^{\top}y+\tilde{v}+\sigma \tilde{\beta}.
\end{align}
In addition, (\ref{pri1}) and (\ref{pri2}) imply
$$d^* \in \lambda \partial \|\cdot\|_1(\beta^*) \Leftrightarrow \beta^*=Prox_{\lambda \|\cdot\|_1}(\beta^*+d^*)=S_{\lambda}(\beta^*+d^*).$$
\end{proof}

Based on Lemma \ref{lemma1}, we can directly apply the PDASC method to solve problem (\ref{linemode2}), which is exhibited in Algorithm \ref{algorithm2}.
%\newpage
\begin{framed}\label{algorithm2}
\noindent
{\bf PDASC method with $(\sigma,\tilde{\beta},\tilde{v})\in \mathbb{R}_{++}\times\mathbb{R}^{p}\times \mathbb{R}^p$}
\vskip 1.0mm \hrule \vskip 1mm
\noindent
\begin{itemize}
\item[Step 0.] Given $\lambda_0\geq\|X^{\top}y\|_{\infty}$, the active set $\A(\lambda_0)=\emptyset$, $\beta(\lambda_0)=\textbf{0}$, $d(\lambda_0)=X^{\top}y$, $\mu\in (0,1)$, $K_{max}\in \mathbb{N}$.
For $j=0,1,\ldots$, do the following operations iteratively.
\item[Step 1.] Let $\lambda_j=\mu \lambda_{j-1}$, $\A_0=\A(\lambda_{j-1})$, $(\beta^0,d^0)=(\beta(\lambda_{j-1}),d(\lambda_{j-1}))$. For $k=1,2,\ldots,K_{max}$, do the following operations iteratively.
\begin{itemize}
\item[Step 1.1.] Compute the active and inactive sets $\A_k$ and $\I_k$:
\begin{align}\label{acti}
\A_k^{+}&=\{i\in [p]: \beta_i^{k-1}+d_{i}^{k-1}>\lambda\},\notag\\
\A_k^{-}&=\{i\in [p]: \beta_i^{k-1}+d_{i}^{k-1}<-\lambda\},\\
\A_{k}&=\A_k^{+} \cup \A_k^{-}, \quad \I_k=\A_{k}^{c}.\notag
\end{align}
\item[Step 1.2.] Check stopping criterion $\A_k=\A_{k-1}$.
\item[Step 1.3.] Update the primal and dual variables $\beta^k$ and $d^k$ respectively by
\begin{align}\label{pd}
\beta_{\I_k}^k&=\textbf{0}_{\I_k}, \quad d_{\A_k}^k=\lambda [\textbf{1}_{\A_k^{+}};-\textbf{1}_{\A_k^{-}}],\notag\\
\beta_{\A_k}^k&=(X_{\A_k}^{\top}X_{\A_k}+\sigma I_{\A_k})^{-1}(X_{\A_k}^{\top}y+\tilde{v}_{\A_k}+\sigma \tilde{\beta}_{\A_k}-d_{\A_k}^k),\\
d_{\I_k}^k&=X_{\I_k}^{\top}y+\tilde{v}_{\I_k}+\sigma \tilde{\beta}_{\I_k}-X_{\I_k}^{\top}X_{\A_k}\beta_{\A_k}^k.\notag
\end{align}
\end{itemize}
\item[Step 2.] Set $\tilde{k}=\min(K_{max}, k)$, $\A(\lambda_j)=\{i\in [p]: \beta_i^{\tilde{k}}+d_i^{\tilde{k}}>\lambda\}\cup \{i\in [p]: \beta_i^{\tilde{k}}+d_i^{\tilde{k}}<-\lambda\}$ and $(\beta(\lambda_j),d(\lambda_j))=(\beta^{\tilde{k}},d^{\tilde{k}})$.
\item[Step 3.] Check stop condition, if stop, employ the high-dimensional Bayesian information criterion (HBIC) to choose the optimal regularization parameter $\hat{\lambda}$ and denote the corresponding $\beta(\hat{\lambda})$ by $\hat{\beta}$. Else, $j:=j+1$.
\end{itemize}
\end{framed}

\begin{remark}
For the step 10 in Algorithm \ref{algorithm2}, the high-dimensional Bayesian information criterion (HBIC) \cite{WKL2013} chooses the optimal $\hat{\lambda}$ by
\begin{align*}
\hat{\lambda}=\argmin_{\lambda\in [\lambda_{min}, \lambda_{max}]}\Big\{HBIC(\lambda):=log\big(\frac{1}{n}\|X\beta(\lambda)-y\|^2\big)+\frac{log(log(n))log(p)}{n}\|\beta(\lambda)\|_0\Big\},
\end{align*}
\end{remark}
where $\lambda_{min}$ and $\lambda_{max}$ will be specified in numerical tests.

\section{Convergence analysis}\label{the}
We firstly describe the convergence result of the algorithm in our first stage. Since $J(\beta;\sigma^1,\textbf{0},\textbf{0})$ is bounded below, we can get the following result from \cite[Proposition 4.19]{F2011} and \cite[Theorem 4.2]{TWST2020}.
\begin{theorem}
Let $\bar{J}(\sigma^1):=\min\limits_{\beta\in\mathbb{R}^{p}}\{J(\beta;\sigma^1,\textbf{0},\textbf{0})\}$. Then we have
\begin{align*}
\lim_{\sigma^1\rightarrow 0}\bar{J}(\sigma^1)=\min_{\beta\in\mathbb{R}^{p}}\Big\{\frac{1}{2}\|X\beta-y\|^2+\lambda \|\beta\|_1\Big\}.
\end{align*}
\end{theorem}

\begin{proof}
For any $\sigma^1>0$ and $\beta\in\mathbb{R}^{p}$, we have
$$\bar{J}(\sigma^1)\leq \frac{1}{2}\|X\beta-y\|^2+\lambda \|\beta\|_1+\frac{\sigma^1}{2}\|\beta\|^2.$$
Therefore, $\lim_{\sigma^1\rightarrow 0}\bar{J}(\sigma^1)\leq \frac{1}{2}\|X\beta-y\|^2+\lambda \|\beta\|_1$. Combining with the arbitrariness of $\beta$, we can get
$$\lim_{\sigma^1\rightarrow 0}\bar{J}(\sigma^1)\leq\min_{\beta\in\mathbb{R}^{p}}\Big\{\frac{1}{2}\|X\beta-y\|^2+\lambda \|\beta\|_1\Big\}.$$
In addition, since $\frac{\sigma^1}{2}\|\beta\|^2\geq 0$ for any $\beta\in\mathbb{R}^{p}$, so
$$\bar{J}(\sigma^1)\geq\min_{\beta\in\mathbb{R}^{p}}\Big\{\frac{1}{2}\|X\beta-y\|^2+\lambda \|\beta\|_1\Big\},$$
and then
$$\lim_{\sigma^1\rightarrow 0}\bar{J}(\sigma^1)\geq\min_{\beta\in\mathbb{R}^{p}}\Big\{\frac{1}{2}\|X\beta-y\|^2+\lambda \|\beta\|_1\Big\}.$$
Hence, we can get the desired result.
\end{proof}

Then, we will analyze the convergence of PMM algorithm. Denote
$$J_k(\beta):=J(\beta;\sigma^{2,k},\beta^k,\nabla q(\beta^k)).$$
At the $k$-th iteration of stage II, we have that
\begin{align}\label{Jk1}
\beta^{k+1}=\argmin_{\beta\in \mathbb{R}^{p}}\{J_k(\beta)+\langle \delta^k, \beta-\beta^k\rangle\}
\end{align}
such that condition (\ref{err1}) is satisfied. The following lemma shows the descent property of the
function $J_k$.

\begin{lemma}\label{le41}
Let $\beta^{k+1}$ be an approximate solution of the subproblem in the $k$-th iteration such
that (\ref{err1}) holds. Then we have
$$J_{k}(\beta^k)\geq J_{k}(\beta^{k+1})-\frac{\sigma^{2,k}}{4}\|\beta^{k+1}-\beta^{k}\|^2.$$
\end{lemma}
\begin{proof}
From the convexity of function $J_{k}$, we have $J_{k}(\beta^k) -J_{k}(\beta^{k+1})\geq \langle \partial J_{k}(\beta^{k+1}) , \beta^{k}-\beta^{k+1}\rangle$. In addition, we can get $-\delta^k \in \partial J_{k}(\beta^{k+1})$ from (\ref{Jk1}). Therefore, we obtain
$$J_{k}(\beta^k) -J_{k}(\beta^{k+1})\geq \langle \delta^k , \beta^{k+1}-\beta^{k}\rangle\geq -\|\delta^k\|\cdot \|\beta^{k+1}-\beta^{k}\|.$$
Combining with condition (\ref{err1}), it is easy to get the desired result
$$J_{k}(\beta^k)\geq J_{k}(\beta^{k+1})-\frac{\sigma^{2,k}}{4}\|\beta^{k+1}-\beta^{k}\|^2.$$
\end{proof}

Next we recall the equivalent expression of a d-stationary point of (\ref{dcmodel}) in the following lemma, which is similar to that in \cite{CPS2018,PRA2016,TWST2020}.

\begin{lemma}\label{d1}
The vector $\bar{\beta}\in \mathbb{R}^{p}$ is a d-stationary point of (\ref{dcmodel}) if and only if there exist $\sigma>0$ such that
$$\bar{\beta}\in \argmin_{\beta\in \mathbb{R}^{p}}\{J(\beta;\sigma,\bar{\beta},\nabla q(\bar{\beta}))\}.$$
\end{lemma}
\begin{proof}
The proof is similar to \cite[Lemma 4.2]{TWST2020}, so it is omitted here.
\end{proof}

Now we present the main result of this section on the subsequential convergence of $\{\beta^k\}$ to a d-stationary point of (\ref{dcmodel}).

\begin{theorem}\label{th1}
Assume $\{\sigma^{2,k}\}$ is a convergent sequence. Let $\{\beta^k\}$ be the sequence generated by
the PMM algorithm. The following two statements hold.\\
1. The function sequence $\{f(\beta^k)\}$ is convergent, and $\lim_{k\rightarrow \infty}\|\beta^{k+1}-\beta^{k}\|=0$.\\
2. Every accumulation point of the sequence $\{\beta^k\}$, if exists, is a d-stationary point of (\ref{dcmodel}).
\end{theorem}
\begin{proof}
1. We can easily get $f(\beta^k)=J_{k}(\beta^k)$. Then from Lemma \ref{le41}, we have
\begin{align*}
f(\beta^k)&=J_{k}(\beta^k)\geq J_{k}(\beta^{k+1})-\frac{\sigma^{2,k}}{4}\|\beta^{k+1}-\beta^{k}\|^2\\
&=\frac{1}{2}\|X \beta^{k+1}-y\|^2+\lambda \|\beta^{k+1}\|_1-q(\beta^k)-\langle \nabla q(\beta^k) ,\beta^{k+1}-\beta^{k} \rangle+\frac{\sigma^{2,k}}{4}\|\beta^{k+1}-\beta^{k}\|^2\\
&=f(\beta^{k+1})+\frac{\sigma^{2,k}}{4}\|\beta^{k+1}-\beta^{k}\|^2+q(\beta^{k+1})-q(\beta^k)-\langle \nabla q(\beta^k) ,\beta^{k+1}-\beta^{k} \rangle\\
&\geq f(\beta^{k+1})+\frac{\sigma^{2,k}}{4}\|\beta^{k+1}-\beta^{k}\|^2.
\end{align*}
The last inequality is derived from the convexity of $q$. Therefore the sequence $\{f(\beta^k)\}$ is non-increasing.
Since $f(\beta)$ is bounded below, the sequence $\{f(\beta^k)\}$ converges, and then the sequence $\{\|\beta^{k+1}-\beta^{k}\|\}$ converges to zero.

2. Let $\beta^{\infty}$ be the limit of a convergent subsequence $\{\beta^k\}_{k\in K_0}$. We can easily prove that $\{\beta^{k+1}\}_{k\in K_0}$ also converges to $\beta^{\infty}$. From the definition of $\beta^{k+1}$, we can get
$$J_{k}(\beta)+\langle \delta^k , \beta-\beta^{k}\rangle\geq J_{k}(\beta^{k+1})+\langle \delta^k , \beta^{k+1}-\beta^{k}\rangle, \quad \forall \beta\in \mathbb{R}^{p}.$$
Further,
\begin{align*}
J_{k}(\beta)\geq J_{k}(\beta^{k+1})+\langle \delta^k , \beta^{k+1}-\beta\rangle
\geq J_{k}(\beta^{k+1})-\|\delta^k\|\cdot\|\beta^{k+1}-\beta\|, \quad \forall \beta\in \mathbb{R}^{p}.
\end{align*}
Letting $k (\in K_0)\rightarrow \infty$, we obtain that $J_{\infty}(\beta)\geq J_{\infty}(\beta^{\infty}), \forall \beta\in \mathbb{R}^{p}$. Equivalently,
$$J(\beta;\sigma^{2,\infty},\beta^{\infty},\nabla q(\beta^{\infty}))\geq J(\beta^{\infty};\sigma^{2,\infty},\beta^{\infty},\nabla q(\beta^{\infty})), \forall \beta\in \mathbb{R}^{p},$$
where $\sigma^{2,\infty}=\lim_{k\rightarrow \infty}\sigma^{2,k}\geq 0$. Then, we can conclude
$$\beta^{\infty}\in \argmin_{\beta\in \mathbb{R}^{p}}\{J(\beta;\sigma^{2,\infty},\beta^{\infty},\nabla q(\beta^{\infty}))\}.$$
From Lemma \ref{d1}, we can easily obtain the desired result.
\end{proof}

%%%%%%%%%%%%%%%%%%%%%%%%%%%%%%%%%%%%%%%%%%%%%%%%%%%%%%%%%%%%%%%%%%%%%%%%%%%%%%%%%%%%%%%%%%%%%%%%%%%%%%%%%%%%%%%%%%%%%%%%
%%%%%%%%%%%%%%%%%%%%%%%%%%%%%%%%%%%%%%%%%%%%%%%%%%%     数值实验   %%%%%%%%%%%%%%%%%%%%%%%%%%%%%%%%%%%%%%%%%%%%%%%%%%%%%
\section{Numerical Experiments}\label{num}
%%%%%%%%%%%%%%%%%%%%%%%%%%%%%%%%%%%%%%%%%%%%%%%%%%%%%%%%%%%%%%%%%%%%%%%%%
In this section, we will use multiple sets of simulated and real examples to illustrate the performance of the proposed PMM algorithm for non-convex penalized high-dimension linear regression problems. The specific layout is that we first use some examples to illustrate the behavior of PMM algorithm, and then highlight the effectiveness and comparability through numerical comparison with the latest SSN method in \cite{SHJ2018} and the classic CD algorithm in \cite{BH2011}. All the experiments are performed with Microsoft Windows 10 and MATLAB R2019a, and run on a PC with an Intel Core i7-9700 CPU at 3.00 GHz and 16 GB of memory.

\subsection{Experiments setting}\label{es1}
In the simulation experiments, we generate the $n\times p$ matrix $X$ whose rows are drawn independently from $\mathcal{N}(0,\Sigma)$ with $\Sigma_{ij}=\kappa^{|i-j|}, 1\leq i,j \leq p$, where $\kappa$ is the correlation coefficient of matrix $X$. In order to generate the target regression coefficient $\beta^*\in \mathbb{R}^{p}$, we randomly select a subset of $\{1,\cdots,p\}$ to form the active set $\A^*$ with $|\A^*|=K<n$. Let $R=m_2 / m_1$, where $m_2=\|\beta^*_{\A^*}\|_{max}$ and $m_1=\|\beta^*_{\A^*}\|_{min}=1$. Then the $K$ nonzero coefficients in $\beta^*$ are uniformly distributed in $[m_1,m_2]$. The response variable is generated by $y=X\beta^*+\varepsilon$ where $\varepsilon \in \mathbb{R}^{n}$ is the additive Gaussian noise and generated independently from $\mathcal{N}(0,\sigma_1^2 I_n)$.

To select the optimal regularization parameter, we set $\lambda_{max}=\|X^{\top}b\|_{\infty}$ and $\lambda_{min}=10^{-10}\lambda_{max}$. Then an equal-distributed partition on log-scale is employed to divide the interval $[\lambda_{min},\lambda_{max}]$ into $N=100$ subintervals.
For the parameter $\tau$, unless otherwise specified, we set $\tau=2.7$ and $\tau=3.7$ for the MCP and SCAD penalties, respectively. Due to the locally one step convergence of the PDAS method for $\ell_1$ regularized least squares problems, we set $K_{max}=1$. And we use the following two relative KKT residuals $R^1_{kkt}$ and $R^2_{kkt}$ to measure the accuracy of the approximate optimal solutions in different stages,
\begin{align}\label{kkt1}
R^1_{kkt}&:=\frac{\Big\|\beta-\prox_{\lambda \|\cdot\|_1}\Big(\beta-X^{\top}(X \beta -y)\Big)\Big\|}{1+\|\beta\|+\|X^{\top}(X \beta-y)\|},\\
R^2_{kkt}&:=\frac{\Big\|\beta-\prox_{\lambda \|\cdot\|_1-q(\cdot)}\Big(\beta-X^{\top}(X \beta -y)\Big)\Big\|}{1+\|\beta\|+\|X^{\top}(X \beta-y)\|},
\end{align}
where the closed form of $\prox_{\lambda \|\cdot\|_1-q(\cdot)}$ can refer to \cite{GZL2013}. Then the PDASC method for solving the internal subproblems is terminated if $R^1_{kkt}<1e-6$, and the PMM algorithm will be terminated if $R^2_{kkt}<1e-6$. In addition, we fix some low-impact parameters, such as $\sigma^1=\sigma^{2,0}=\gamma=0.1$. The values of other parameters will be given in the context of specific issues.

In addition, for the purpose of highlighting the efficiency and accuracy of PMM algorithm in the subsequent simulation comparison, we compare it with the latest SSN method and the classic CD algorithm from the perspective of the following four indicators based on 100 independent experiments:
\begin{itemize}
  \item The average CPU time (Time, in seconds);
  \item The average $\ell_2$ relative error: $RE:=\frac{\sum_{m=1}^{100}\big(\frac{\|\hat{\beta}^{(m)}-\beta^*\|_2}{\|\beta^*\|_2}\big)}{100}$;
  \item The average estimated model size: $MS:=\frac{\sum_{m=1}^{100}|\hat{\A}^{(m)}|}{100}$;
  \item The proportion of correct models: $CM:=\frac{\sum_{m=1}^{100}\delta\{\hat{\A}^{(m)}=\A^*\}}{100}$,
\end{itemize}
where $\hat{\beta}$ and $\hat{\A}$ are the estimated regression coefficient and active set, respectively. $|\A|$ indicates the length of set $\A$, and $\delta\{\hat{\A}^{(m)}=\A^*\}=\left\{
\begin{array}{ll}
1, &\hat{\A}^{(m)}=\A^*\\
0, &\hat{\A}^{(m)}\neq \A^*
\end{array}
\right.$. Clearly, the smaller Time, the faster calculation speed. And the smaller RE, the closer MS approaches to $K$, the closer CM approaches to $100\%$, the higher the solution quality.

\subsection{The behavior of PMM algorithm}

In this part, we analyze the computational behavior of the PMM algorithm based on 100 independent experiments and consider the problem setting with $n=300$, $p=1000$, $K=10$, $\sigma_1=0.1$, $R=100$. Here we only give the results related to the MCP penalty, since SCAD penalty will produce a similar phenomenon.

Firstly, we utilize a box plot to investigate the performance of variable selection and parameter estimation for the PMM algorithm. To achieve the goal, we generate a coefficient matrix $X$ with $\kappa=0.2$ and a fixed true regression parameter $\beta^*$, whose 10 non-zero elements are $\beta^*_{30}=6$, $\beta^*_{198}=-11$, $\beta^*_{269}=-10$, $\beta^*_{395}=25$, $\beta^*_{442}=-8$, $\beta^*_{495}=100$, $\beta^*_{637}=-9$, $\beta^*_{766}=-10$, $\beta^*_{777}=5$, $\beta^*_{865}=1$. In view of the large $p$, we only describe the estimation effect of non-zero elements in $\beta^*$ on the left side of Figure \ref{fig:1}. Obviously, for each non-zero element, the estimated results fluctuate very little in 100 independent experiments, which fully illustrates the effectiveness and stability of the PMM algorithm. In addition, the private experiment shows that the positions which should be zero are all 0. Therefore, we conclude that the PMM algorithm can simultaneously realize variable selection and parameter estimation.

Next, we examine the calculation speed of PMM algorithm from the perspective of the number of iterations. Based on 100 independent experiments, we show the average number of iterations with different sparsity levels on the right side of Fig. \ref{fig:1}. In view of the stop condition $\|\beta(\lambda_j)\|_0\geq n/log(p)$ in step 10 of PDASC method, here we consider $K=5:5:40$, which means that the sparsity level $K$ varies from 5 to 40 by step 5. In addition, we also take the correlation into consideration and set $\kappa=[0.3;0.5;0.7]$. It can be seen that for the three correlation coefficients, the average number of iterations of the PMM algorithm does not exceed 4. This phenomena fully illustrates that the calculation speed of the PMM algorithm is very fast.

\begin{figure}[h]
\centering
\includegraphics[width=0.5\textwidth]{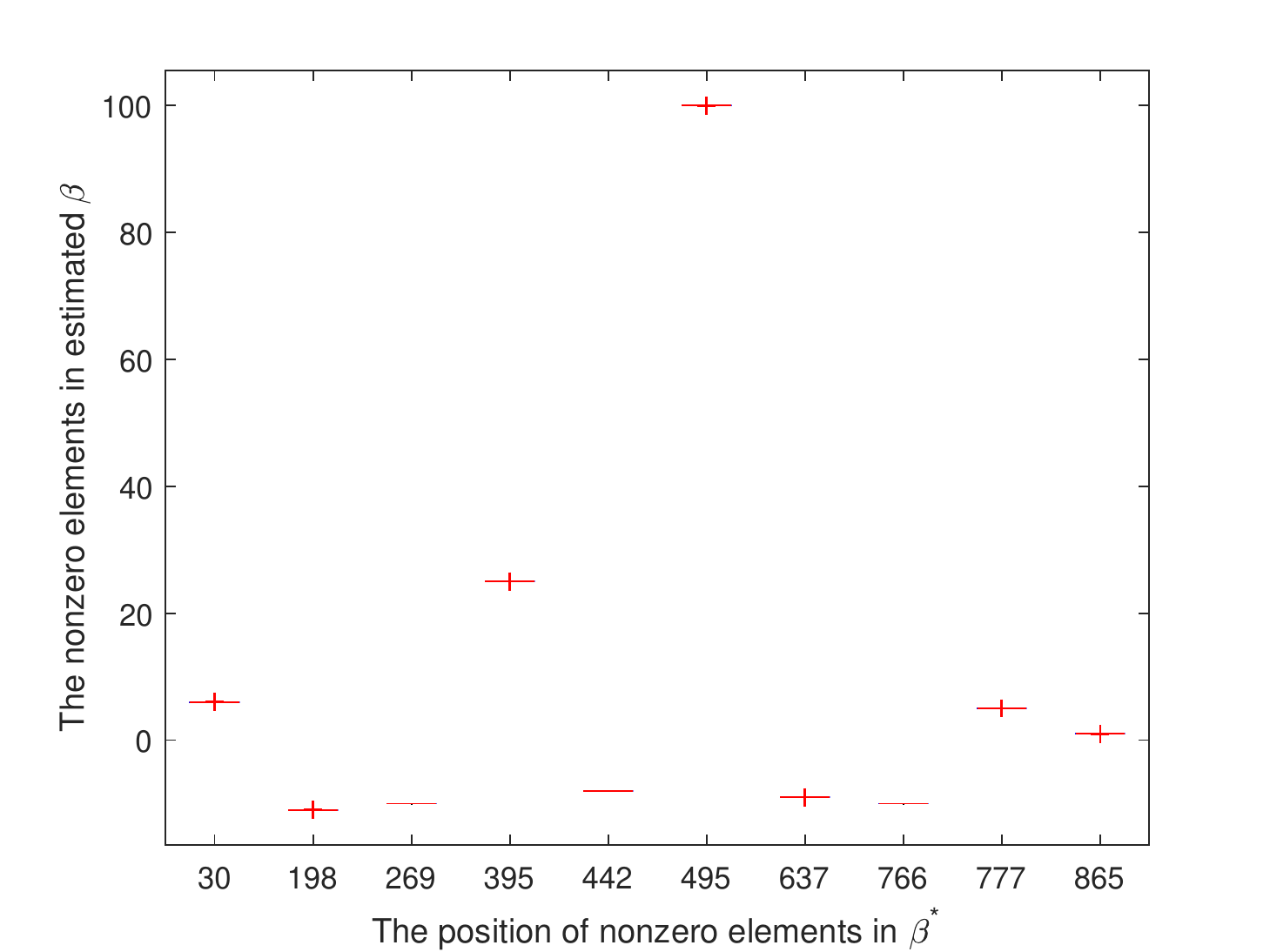}\hspace{-.4cm}
\includegraphics[width=0.5\textwidth]{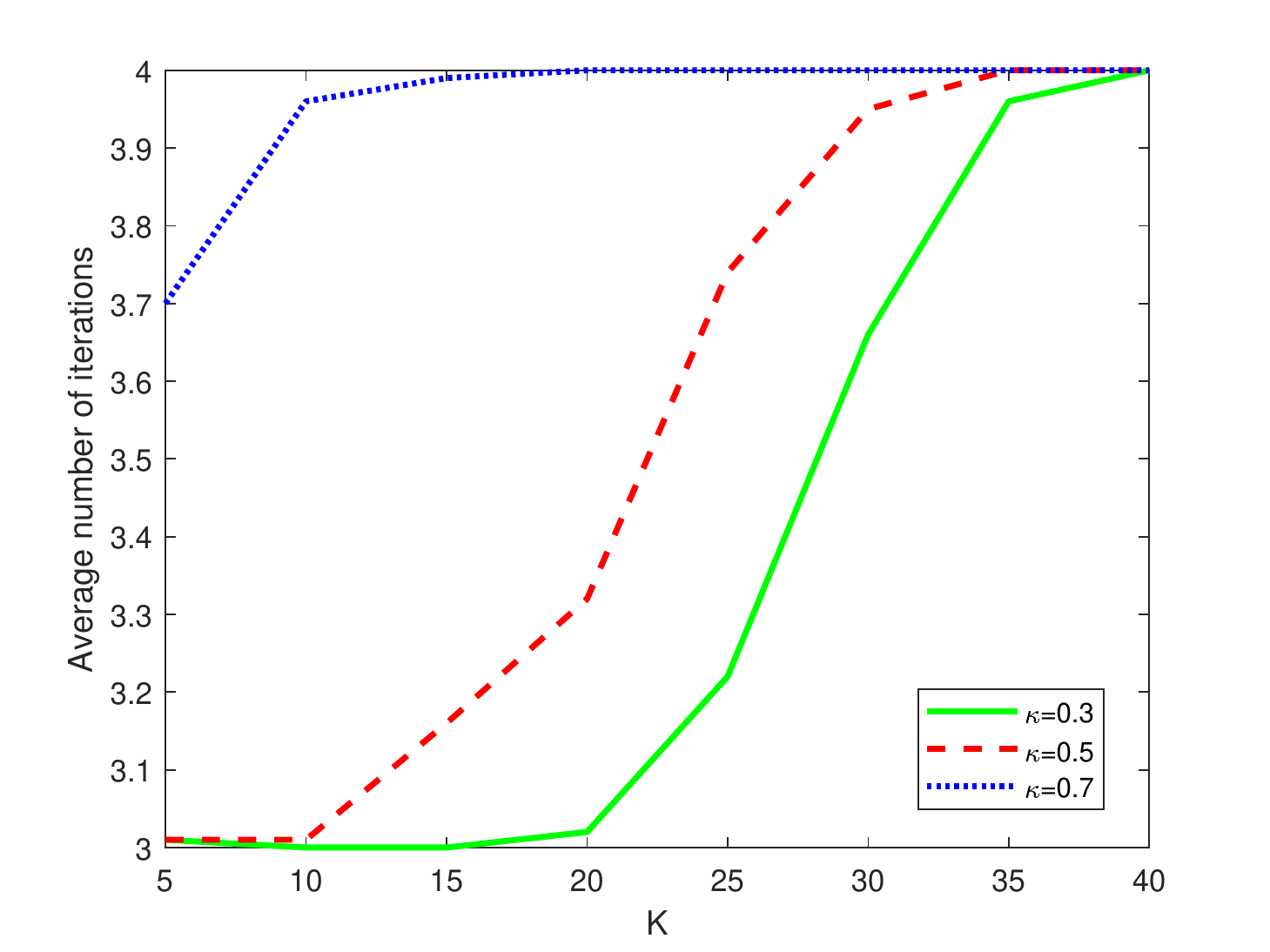}
\caption{{\small The behavior of PMM algorithm for MCP penalized linear regression problems based on 100 independent experiments}}
\label{fig:1}
\end{figure}

\subsection{Comparison with SSN algorithm}

In this part, we compare the PMM algorithm with the latest SSN algorithm for solving non-convex penalized high-dimensional linear regression problems based on 100 independent experiments. We set $p=2000$ with $n=\lfloor \frac{p}{5}\rfloor$ and $K=\lfloor \frac{n}{2 log(p)}\rfloor$, where $\lfloor x \rfloor$ denotes the integer part of $x$ for $x\geq 0$. Here, we set $\sigma_1=0.1$ and consider three levels of correlation, i.e., $\kappa=[0.3;0.5;0.7]$. It can be observed from the MATLAB package of SSN algorithm that the authors in \cite{SHJ2018} lead into a key parameter ``Weight" which represents the step size in the programming process. After testing, we find that the effectiveness of SSN algorithm is heavily dependent on this parameter. Here we only consider two values of 0.5 and 0.9. In addition, for the sake of fairness, we use the same continuation method for the regularization parameter in SSN algorithm, unify the maximum number of iterations to 1, and other parameters are consistent with their original papers. Simulation results are summarized in Table \ref{tab01}.

%\begin{table}[!ht]
%\centering {\small\caption{Simulation results of SSN, PDAS and PMM algorithms}
%\begin{tabular}{|c|c|c|c|c|c|c|c|}
%\hline
%p & $\kappa$  & Penalty & Method & Time & MS & CM & RE         \\
%\hline
%\multirow{18}{*}{2000} & \multirow{6}{*}{0.3}  & \multirow{3}{*}{MCP} & PMM & 0.07 & 26.00 & 100\% & 1.48e-4   \\
%& & &  SSN & 0.04 & 26.01 & 99\% & 1.00e-4 \\
%& & & UPDASC & 0.02 & 26.05 & 97\% & 8.00e-4   \\
%\cline{3-8}
%& & \multirow{3}{*}{SCAD} & PMM & 0.06 & 26.00 & 100\% & 1.48e-4   \\
%& & & SSN & 0.05 & 26.00 & 100\% & 1.00e-4  \\
%& & & UPDASC & 0.02 & 26.10 & 96\% & 5.00e-4   \\
%\cline{2-8}
%& \multirow{6}{*}{0.5}  & \multirow{3}{*}{MCP} & PMM & 0.07 & 26.00 & 100\% & 1.41e-4   \\
%& & & SSN & 0.04 & 25.95 & 98\% & 5.00e-4 \\
%& & & UPDASC & 0.01 & 27.50 & 8\% & 6.65e-2   \\
%\cline{3-8}
%& & \multirow{3}{*}{SCAD} & PMM & 0.06 & 26.00 & 100\% & 1.41e-4   \\
%& & & SSN & 0.05 & 25.96 & 99\% & 4.00e-4  \\
%& & & UPDASC & 0.02 & 30.05 & 36\% & 3.17e-2   \\
%\cline{2-8}
%& \multirow{6}{*}{0.7}  & \multirow{3}{*}{MCP} & PMM & 0.09 & 26.00 & 100\% & 1.46e-4   \\
%& & & SSN & 0.03 & 23.94 & 68\% & 6.60e-2 \\
%& & & UPDASC & 0.01 & 43.77 & 0\% & 1.41e-1   \\
%\cline{3-8}
%& & \multirow{3}{*}{SCAD} & PMM & 0.08 & 26.00 & 100\% & 1.46e-4   \\
%& & & SSN & 0.04 & 25.40 & 93\% & 1.50e-2  \\
%& & & UPDASC & 0.01 & 16.15 & 0\% & 6.24e-1   \\
%\hline
%\end{tabular}\label{tab01}
%}
%\end{table}

\begin{table}[!ht]
\centering {\small\caption{Simulation results of SSN and PMM algorithms}
\begin{tabular}{|c|c|c|c|c|c|c|c|}
\hline
$\kappa$ & Weight  & Penalty & Method & Time & MS & CM & RE         \\
\hline
\multirow{8}{*}{0.3} & \multirow{4}{*}{0.5} & \multirow{2}{*}{MCP} & SSN & 0.04 & 26.01 & 99\% & 1.00e-4 \\
& & & PMM & 0.07 & 26.00 & 100\% & 1.48e-4   \\
\cline{3-8}
& & \multirow{2}{*}{SCAD} & SSN & 0.05 & 26.00 & 100\% & 1.00e-4  \\
& & & PMM & 0.07 & 26.00 & 100\% & 1.48e-4   \\
\cline{2-8}
& \multirow{4}{*}{0.9} & \multirow{2}{*}{MCP} & SSN & 0.05 & 26.40 & 78\% & 2.00e-4 \\
& & & PMM & 0.08 & 26.00 & 100\% & 1.48e-4   \\
\cline{3-8}
& & \multirow{2}{*}{SCAD} & SSN & 0.06 & 26.76 & 54\% & 1.73e-4  \\
& & & PMM & 0.07 & 26.00 & 100\% & 1.48e-4   \\
\cline{1-8}
\multirow{8}{*}{0.5} & \multirow{4}{*}{0.5} & \multirow{2}{*}{MCP} & SSN & 0.04 & 25.95 & 98\% & 5.00e-4 \\
& & & PMM & 0.08 & 26.00 & 100\% & 1.41e-4   \\
\cline{3-8}
& & \multirow{2}{*}{SCAD} & SSN & 0.05 & 25.96 & 99\% & 4.00e-4  \\
& & & PMM & 0.07 & 26.00 & 100\% & 1.41e-4   \\
\cline{2-8}
& \multirow{4}{*}{0.9} & \multirow{2}{*}{MCP} & SSN & 0.06 & 26.33 & 80\% & 2.00e-4 \\
& & & PMM & 0.08 & 26.00 & 100\% & 1.41e-4   \\
\cline{3-8}
& & \multirow{2}{*}{SCAD} & SSN & 0.06 & 26.62 & 65\% & 1.60e-4  \\
& & & PMM & 0.07 & 26.00 & 100\% & 1.41e-4   \\
\cline{1-8}
\multirow{8}{*}{0.7} & \multirow{4}{*}{0.5} & \multirow{2}{*}{MCP} & SSN & 0.03 & 23.94 & 68\% & 6.60e-2 \\
& & & PMM & 0.09 & 26.00 & 100\% & 1.46e-4   \\
\cline{3-8}
& & \multirow{2}{*}{SCAD} & SSN & 0.05 & 25.40 & 93\% & 1.50e-2  \\
& & & PMM & 0.09 & 26.00 & 100\% & 1.46e-4   \\
\cline{2-8}
& \multirow{4}{*}{0.9} & \multirow{2}{*}{MCP} & SSN & 0.04 & 24.96 & 60\% & 2.76e-2 \\
& & & PMM & 0.09 & 26.00 & 100\% & 1.46e-4   \\
\cline{3-8}
& & \multirow{2}{*}{SCAD} & SSN & 0.06 & 26.42 & 67\% & 1.61e-4  \\
& & & PMM & 0.09 & 26.00 & 100\% & 1.46e-4   \\
\hline
\end{tabular}\label{tab01}
}
\end{table}

From the information in Table \ref{tab01}, we can see that the calculation speed of SSN algorithm is very fast, which thanks to its local super-linear convergence. However, since its performance is heavily dependent on the selection of the step size, the results under the fixed step sizes $0.5$ and $0.9$ are incomparable with PMM algorithm at present. Therefore, in view of the fact that the SSN algorithm need to carefully adjust the step size under different problem settings, we will only compare the algorithm in this paper with the classic CD algorithm detailly in the subsequent numerical experiments.

\subsection{Comparison with CD algorithm}

In this section, we compare our PMM algorithm with the CD algorithm in \cite{BH2011} for solving (\ref{primal}) which is summarized in Algorithm \ref{algorithm3}. To be fair, we here use the same continuation method for regularization parameter $\lambda$ and the same stop condition $R^2_{kkt}<1e-6$ at step 8. In addition, we also set $K_{max}=1$ to improve the calculation speed of the CD algorithm.

\begin{framed}\label{algorithm3}
\noindent
{\bf CD algorithm}
\vskip 1.0mm \hrule \vskip 1mm
\noindent
\begin{itemize}
\item[Step 0.] Given $\lambda$, $\beta^0=\textbf{0}$, $r^0=y-X\beta^0$, $K_{max}\in \mathbb{N}$.
For $k=0,1,\ldots,K_{max}$, do the following operations iteratively.
\item[Step 1.] For $i=1,2,\ldots,p$, do the following operations iteratively.
\begin{itemize}
\item[Step 1.1.] Calculate $z_i^k=X_i^{\top}r^k+\beta_i^k$, where $X_i$ is the $i$th column of $X$ and $r^k=y-X \beta^k$ is the current residual value.
\item[Step 1.2.] Update
$\beta_i^{k+1}=\prox_{\lambda \|\cdot\|_1-q(\cdot)}(z_i^k)$.
\item[Step 1.3.] Update $r^{k+1}=r^k-(\beta_i^{k+1}-\beta_i^k)X_i$.
\end{itemize}
\item[Step 2.] Check stop condition, if stop, denote the last iteration by $\hat{\beta}$. Else, $k:=k+1.$
\end{itemize}
\end{framed}

\subsubsection{Efficiency and accuracy}

In this part, we compare the efficiency and accuracy of the PMM algorithm and the CD algorithm based on 100 independent experiments. We set $p=2000$ and $5000$ with $n=\lfloor \frac{p}{5}\rfloor$ and $K=\lfloor \frac{n}{2 log(p)}\rfloor$. We consider three levels of correlation ($\kappa=[0.3;0.5;0.7]$) and two levels of noises ($\sigma_1=[0.1;1]$). Simulation results are summarized in Table \ref{tab1}.

\begin{table}[!ht]
\centering {\scriptsize\caption{Simulation results of CD and PMM algorithms}
\begin{tabular}{|c|c|c|c|c|c|c|c|c|}
\hline
p & $\kappa$ & $\sigma_1$ & Penalty & Method & Time & MS & CM & RE         \\
\hline
\multirow{24}{*}{2000} & \multirow{8}{*}{0.3} & \multirow{4}{*}{0.1} & \multirow{2}{*}{MCP} & CD & 0.69 & 26.00 & 100\% & 1.48e-4  \\
& & & & PMM & 0.08 & 26.00 & 100\% & 1.48e-4   \\
& & & \multirow{2}{*}{SCAD} & CD & 0.70 & 26.00 & 100\% & 1.48e-4  \\
& & & & PMM & 0.08 & 26.00 & 100\% & 1.48e-4  \\
\cline{3-9}
& &\multirow{4}{*}{1} & \multirow{2}{*}{MCP} & CD & 0.52 & 26.00 & 100\% & 1.50e-3  \\
& & &  & PMM & 0.07 & 26.00 & 100\% & 1.50e-3  \\
& & & \multirow{2}{*}{SCAD} & CD & 0.52 & 26.01 & 99\% & 1.50e-3 \\
& & &  & PMM & 0.08 & 26.00 & 100\% & 1.50e-3 \\
\cline{2-9}
& \multirow{8}{*}{0.5} & \multirow{4}{*}{0.1} & \multirow{2}{*}{MCP} & CD & 0.70 & 26.00 & 100\% & 1.41e-4  \\
& & & & PMM & 0.08 & 26.00 & 100\% & 1.41e-4\\
& & & \multirow{2}{*}{SCAD} & CD & 0.71 & 26.00 & 100\% & 1.41e-4 \\
& & & & PMM & 0.08 & 26.00 & 100\% & 1.41e-4 \\
\cline{3-9}
& &\multirow{4}{*}{1} & \multirow{2}{*}{MCP} & CD & 0.53 & 26.00 & 100\% & 1.40e-3 \\
& & &  & PMM & 0.07 & 26.00 & 100\% & 1.40e-3 \\
& & & \multirow{2}{*}{SCAD} & CD & 0.54 & 26.00 & 100\% & 1.40e-3 \\
& & &  & PMM & 0.08 & 26.00 & 100\% & 1.40e-3 \\
\cline{2-9}
& \multirow{8}{*}{0.7} & \multirow{4}{*}{0.1} & \multirow{2}{*}{MCP} & CD & 0.70 & 26.00 & 100\% & 1.46e-4 \\
& & & & PMM & 0.10 & 26.00 & 100\% & 1.46e-4 \\
& & & \multirow{2}{*}{SCAD} & CD & 0.71 & 26.00 & 100\% & 1.46e-4 \\
& & & & PMM & 0.10 & 26.00 & 100\% & 1.46e-4 \\
\cline{3-9}
& &\multirow{4}{*}{1} & \multirow{2}{*}{MCP} & CD & 0.53 & 26.05 & 97\% & 1.60e-3 \\
& & &  & PMM & 0.08 & 26.00 & 100\% & 1.50e-3 \\
& & & \multirow{2}{*}{SCAD} & CD & 0.53 & 26.10 & 94\% & 1.60e-3 \\
& & &  & PMM & 0.07 & 26.00 & 100\% & 1.50e-3 \\
\hline
\multirow{24}{*}{5000} & \multirow{8}{*}{0.3} & \multirow{4}{*}{0.1} & \multirow{2}{*}{MCP} & CD & 2.84 & 58.00 & 100\% & 9.29e-5  \\
& & & & PMM & 1.07 & 58.00 & 100\% & 9.29e-5  \\
& & & \multirow{2}{*}{SCAD} & CD & 2.85 & 58.00 & 100\% & 9.29e-5  \\
& & & & PMM & 1.05 & 58.00 & 100\% & 9.29e-5  \\
\cline{3-9}
& &\multirow{4}{*}{1} & \multirow{2}{*}{MCP} & CD & 2.25 & 58.00 & 100\% & 9.29e-4  \\
& & &  & PMM & 1.04 & 58.00 & 100\% & 9.29e-4  \\
& & & \multirow{2}{*}{SCAD} & CD & 2.48 & 58.00 & 100\% & 9.00e-4  \\
& & &  & PMM & 1.03 & 58.00 & 100\% & 9.29e-4  \\
\cline{2-9}
& \multirow{8}{*}{0.5} & \multirow{4}{*}{0.1} & \multirow{2}{*}{MCP} & CD & 2.85 & 58.00 & 100\% & 9.34e-5 \\
& & & & PMM & 1.06 & 58.00 & 100\% & 9.34e-5  \\
& & & \multirow{2}{*}{SCAD} & CD & 2.85 & 58.00 & 100\% & 9.34e-5  \\
& & & & PMM & 1.03 & 58.00 & 100\% & 9.34e-5 \\
\cline{3-9}
& &\multirow{4}{*}{1} & \multirow{2}{*}{MCP} & CD & 2.45 & 58.00 & 100\% & 9.34e-4  \\
& & &  & PMM & 1.10 & 58.00 & 100\% & 9.34e-4  \\
& & & \multirow{2}{*}{SCAD} & CD & 2.49 & 58.00 & 100\% & 9.00e-4  \\
& & &  & PMM & 1.14 & 58.00 & 100\% & 9.34e-4 \\
\cline{2-9}
& \multirow{8}{*}{0.7} & \multirow{4}{*}{0.1} & \multirow{2}{*}{MCP} & CD & 2.87 & 58.00 & 100\% & 9.77e-5  \\
& & & & PMM & 1.31 & 58.00 & 100\% & 9.77e-5  \\
& & & \multirow{2}{*}{SCAD} & CD & 2.87 & 58.00 & 100\% & 9.77e-5  \\
& & & & PMM & 1.29 & 58.00 & 100\% & 9.77e-5 \\
\cline{3-9}
& &\multirow{4}{*}{1} & \multirow{2}{*}{MCP} & CD & 2.41 & 58.03 & 99\% & 9.96e-4  \\
& & &  & PMM & 1.37 & 58.00 & 100\% & 9.77e-4  \\
& & & \multirow{2}{*}{SCAD} & CD & 2.28 & 58.03 & 99\% & 1.00e-3  \\
& & &  & PMM & 1.40 & 58.00 & 100\% & 9.77e-4  \\
\hline
\end{tabular}\label{tab1}
}
\end{table}

From the results of MS, CM and RE in Table \ref{tab1}, it can be concluded that for each combination of $(p,\kappa,\sigma_1)$, the PMM algorithm can always achieve variable selection and parameter estimation very accurately. In addition, the PMM algorithm has better speed performance than CD algorithm for both MCP and SCAD, and PMM is about $2\sim9$ times faster than CD. In particular, for given penalty and method, the CPU time increases with the increase of $p$, and decreases with the increase of $\sigma_1$, but does not change much for different $\kappa$. In addition, it can be found that larger $p$ can improve the accuracy of both CD and PMM, while larger $\sigma_1$ has the opposite effect. Overall, the simulation results in Table \ref{tab1} illustrate that PMM outperforms CD in terms of CPU time while producing solutions of comparable quality.
\subsubsection{Influence of model parameters}

We now consider the effects of each of the model parameters $(n,p,K,\kappa,\sigma_1,\tau)$ on the performance of PMM and CD algorithms. Here we only give the results related to the MCP penalty, since SCAD penalty will produce a similar phenomenon. Based on 10 independent replications, we compare the performance of the considered methods in terms of average positive discovery rate (APDR), average false discovery rate (AFDR) and average combined discovery rate (ACDR)\cite{LC2014} defined as follows:
$$\text{APDR}=\frac{1}{10}\sum \frac{|\hat{\A}\bigcap \A^*|}{|\A^*|},\quad  \text{AFDR}=\frac{1}{10}\sum \frac{|\hat{\A}\bigcap \A^{*c}|}{|\hat{\A}|}, \quad \text{ACDR}=\text{APDR}+(1-\text{AFDR}),$$
where $\A^*$ denotes the true active set and $\A^{*c}$ denotes the complement of $\A^{*}$. Results of APDR, AFDR and ACDR over 10 independent replications are given in Fig. \ref{fig:2}-\ref{fig:4}, respectively. The parameters for solvers are set as follows.

\begin{itemize}
  \item Influence of the sample size $n$: We set $p=1000$, $K=10$, $\tau=2.7$, $\kappa=0.2$, $\sigma_1=0.1$, and take $n=20$ to $200$ with a step size $20$.
  \item Influence of the dimension $p$: We set $n=200$, $K=50$, $\tau=2.7$, $\kappa=0.2$, $\sigma_1=0.1$, and take $p=500$ to $1000$ with a step size $100$.
  \item Influence of the sparsity level $K$: We set $n=200$, $p=1000$, $\tau=2.7$, $\kappa=0.2$, $\sigma_1=0.1$, and take $K=10$ to $50$ with a step size $10$.
  \item Influence of the correlation level $\kappa$: We set $n=200$, $p=1000$, $K=40$, $\tau=2.7$, $\sigma_1=0.1$, and take $\kappa=0.1$ to $0.7$ with a step size $0.1$.
  \item Influence of the noise level $\sigma_1$: We set $n=200$, $p=1000$, $K=10$, $\tau=2.7$, $\kappa=0.2$, and take $\sigma_1\in \{0.1, 0.5, 1.0, 1.5, 2.0, 2.5\}$.
  \item Influence of the concavity parameter $\tau$: We set $n=200$, $p=1000$, $K=10$, $\kappa=0.2$, $\sigma_1=0.1$, and take $\tau\in \{1.1, 2.7, 5, 10\}$.
\end{itemize}

\begin{figure}[!ht]
\centering
\includegraphics[width=0.34\textwidth]{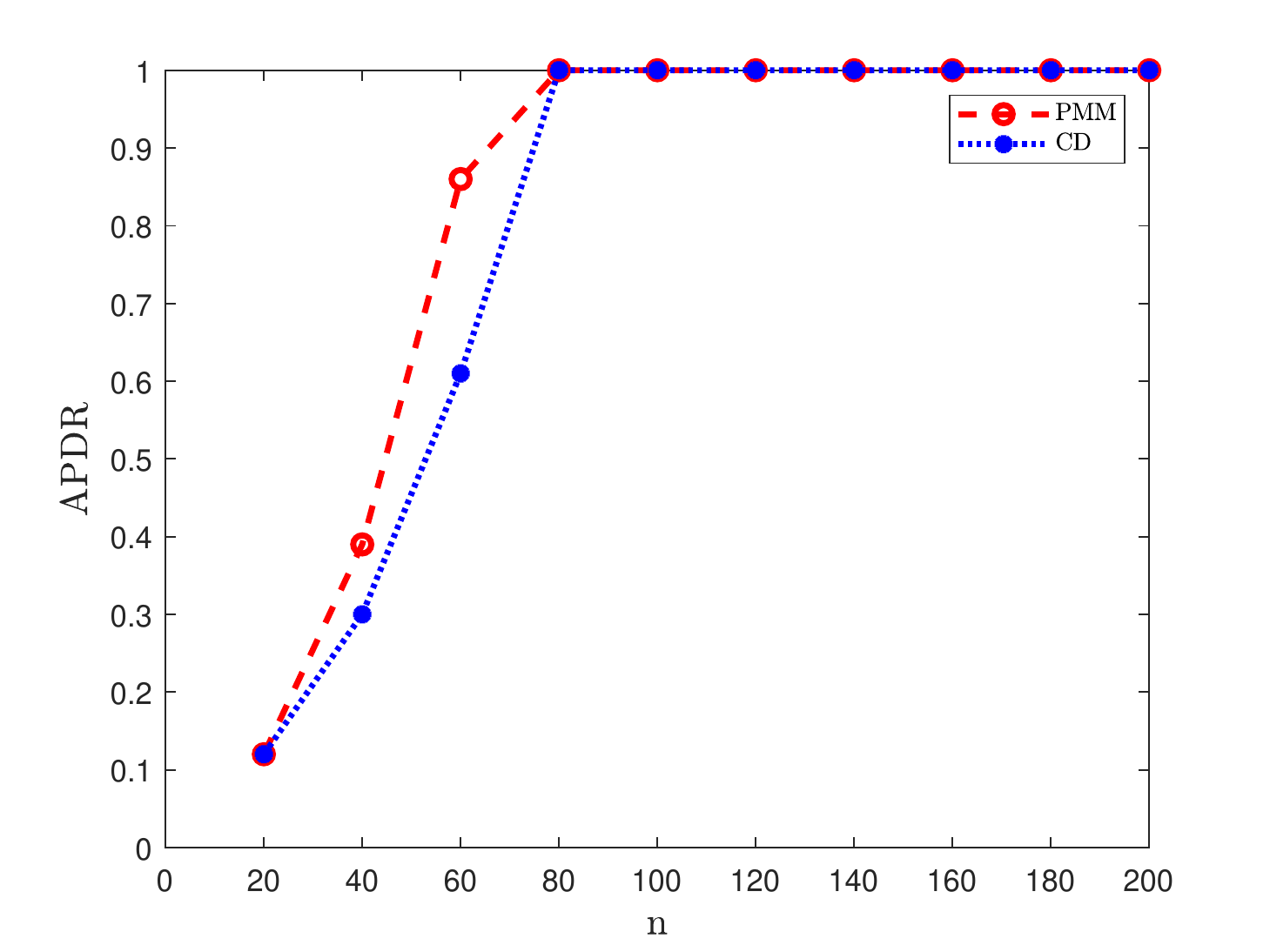}\hspace{-.4cm}
\includegraphics[width=0.34\textwidth]{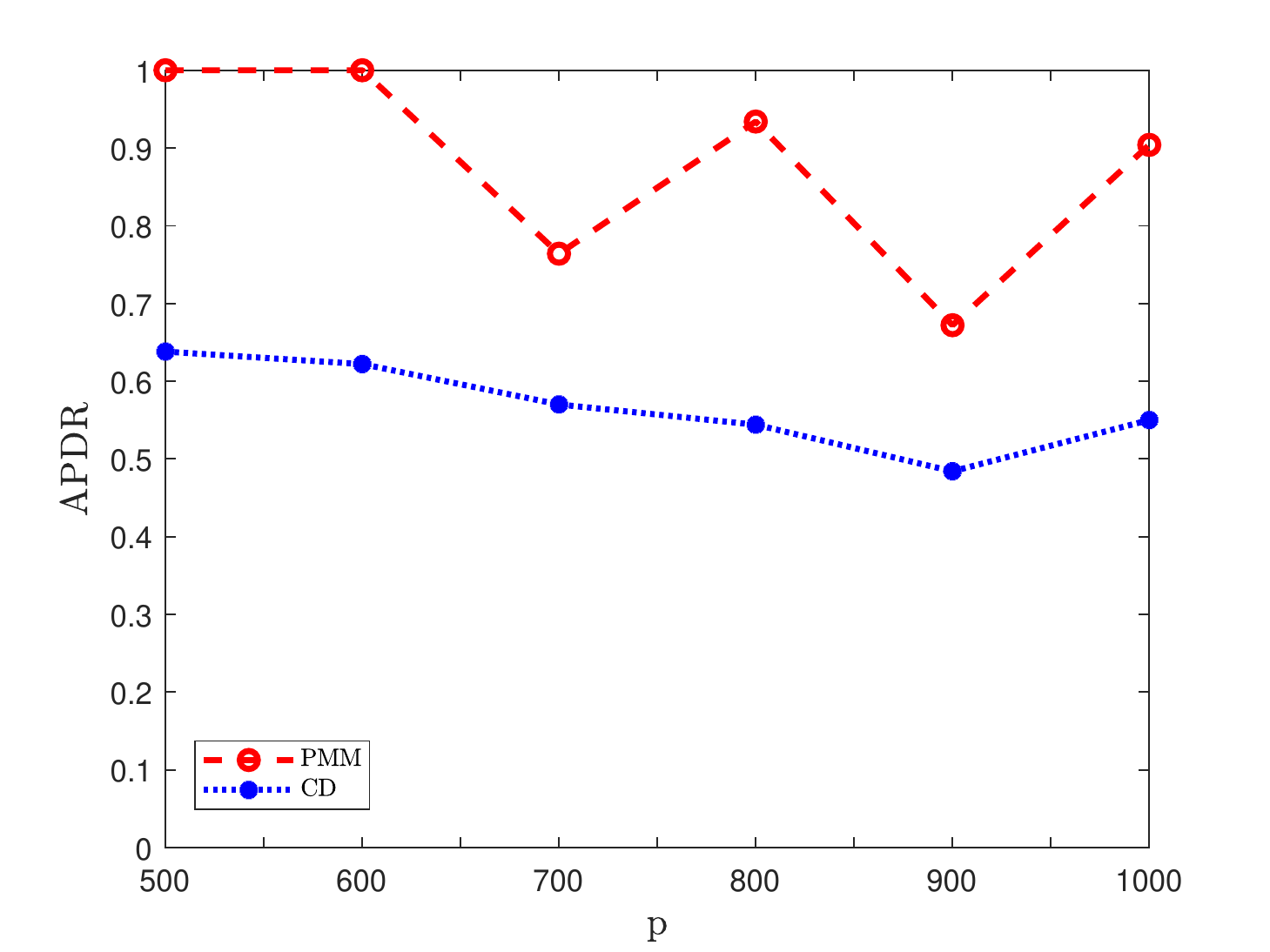}\hspace{-.4cm}
\includegraphics[width=0.34\textwidth]{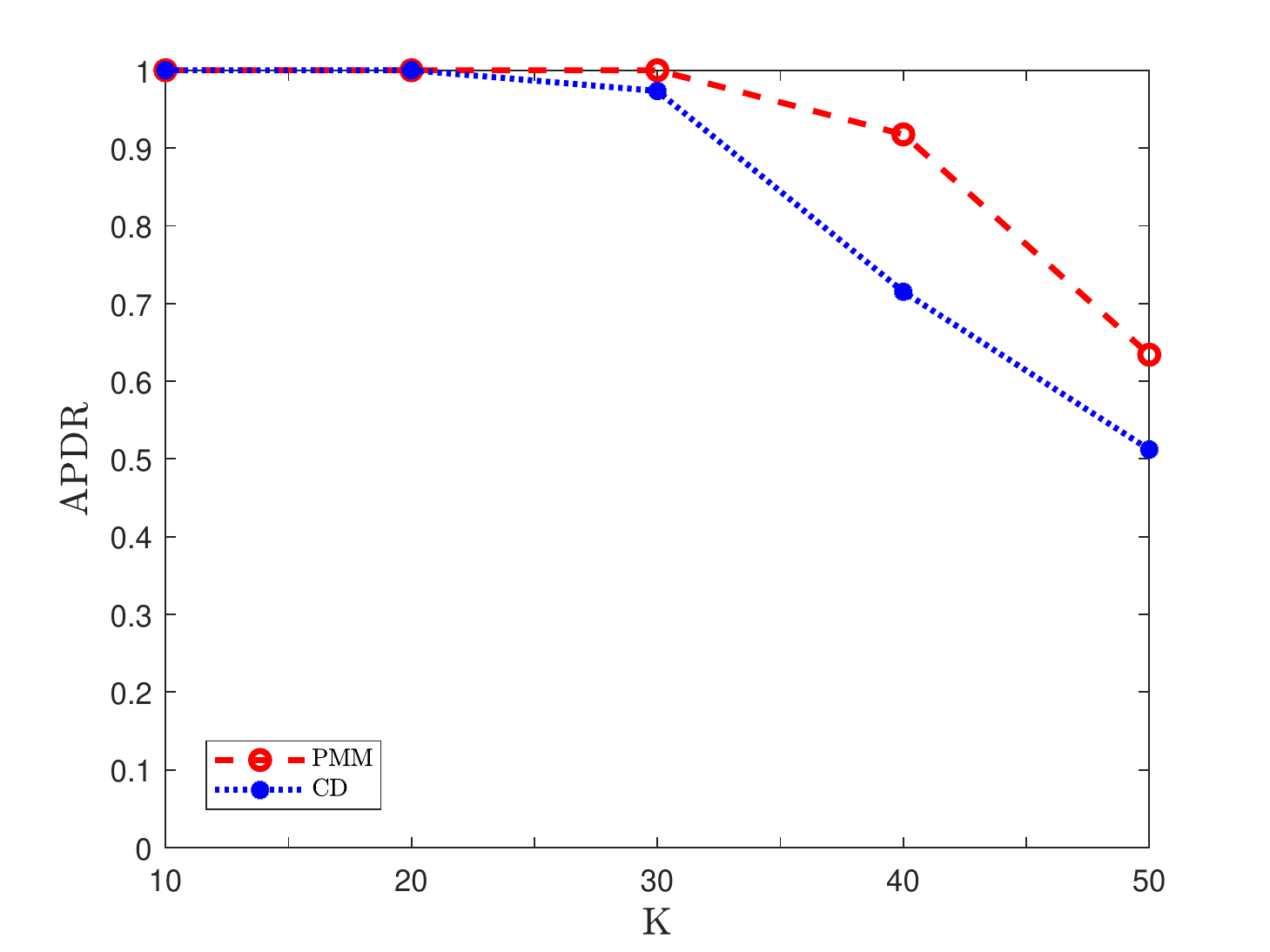}\\
\includegraphics[width=0.34\textwidth]{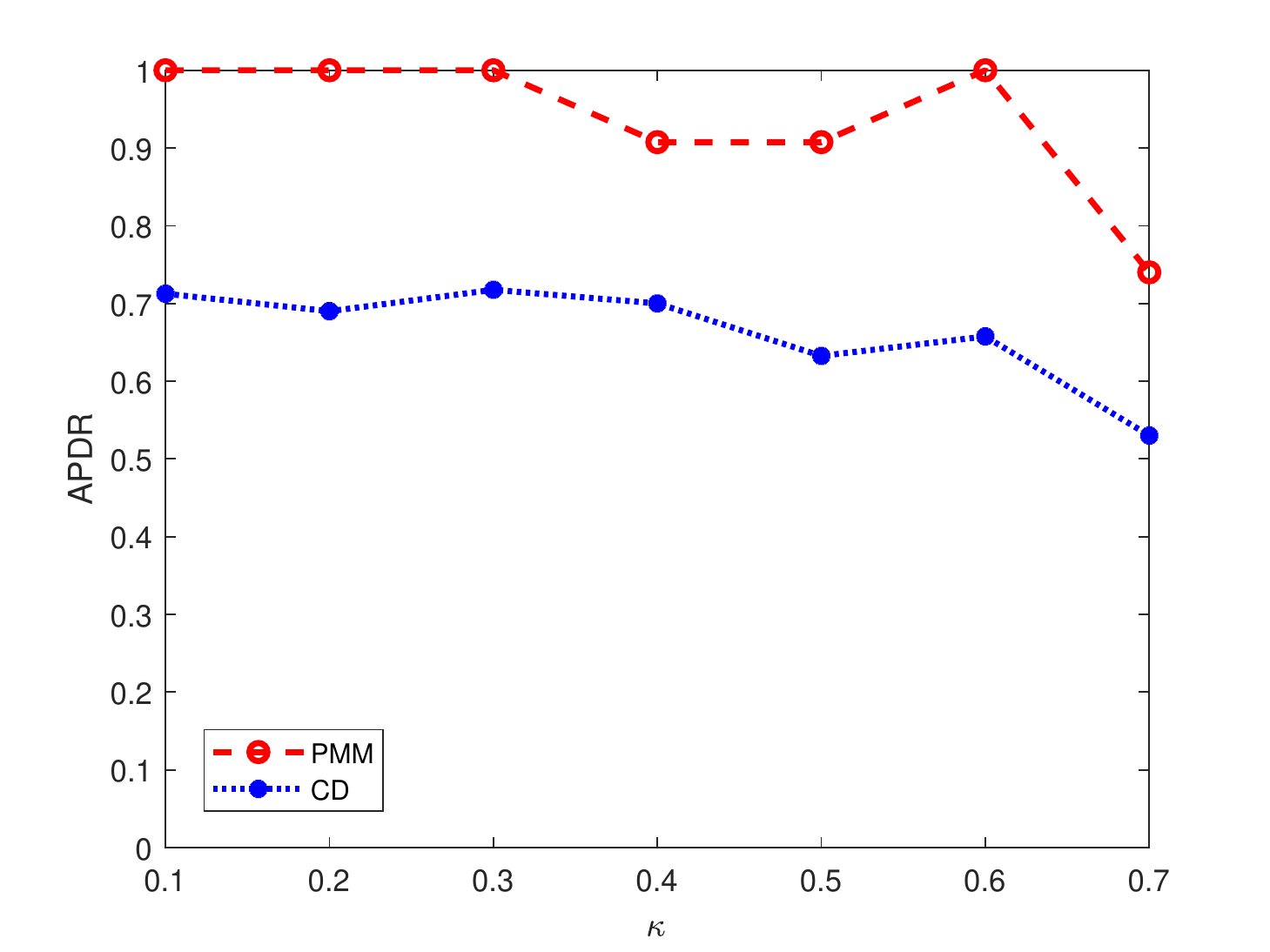}\hspace{-.4cm}
\includegraphics[width=0.34\textwidth]{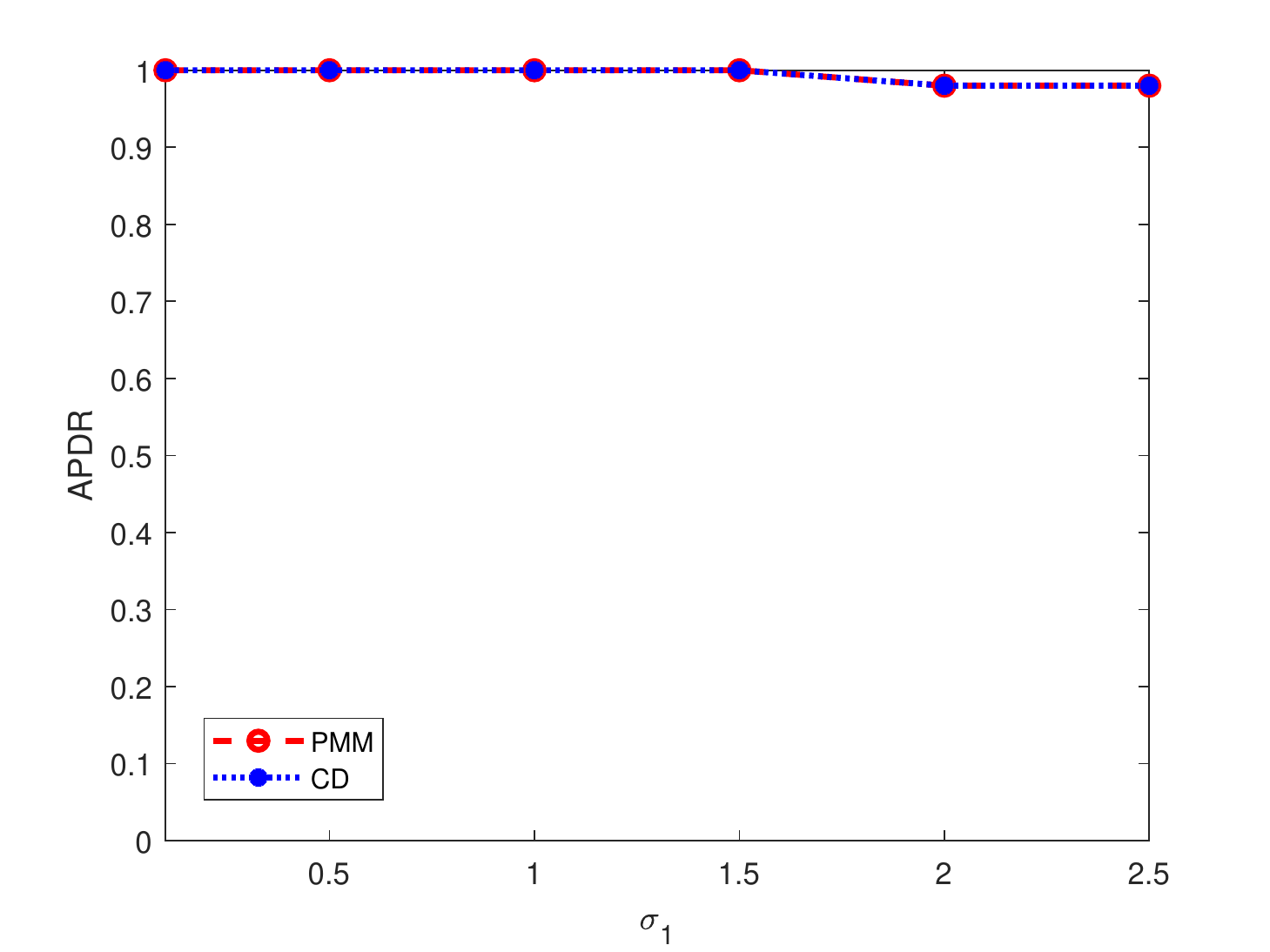}\hspace{-.4cm}
\includegraphics[width=0.34\textwidth]{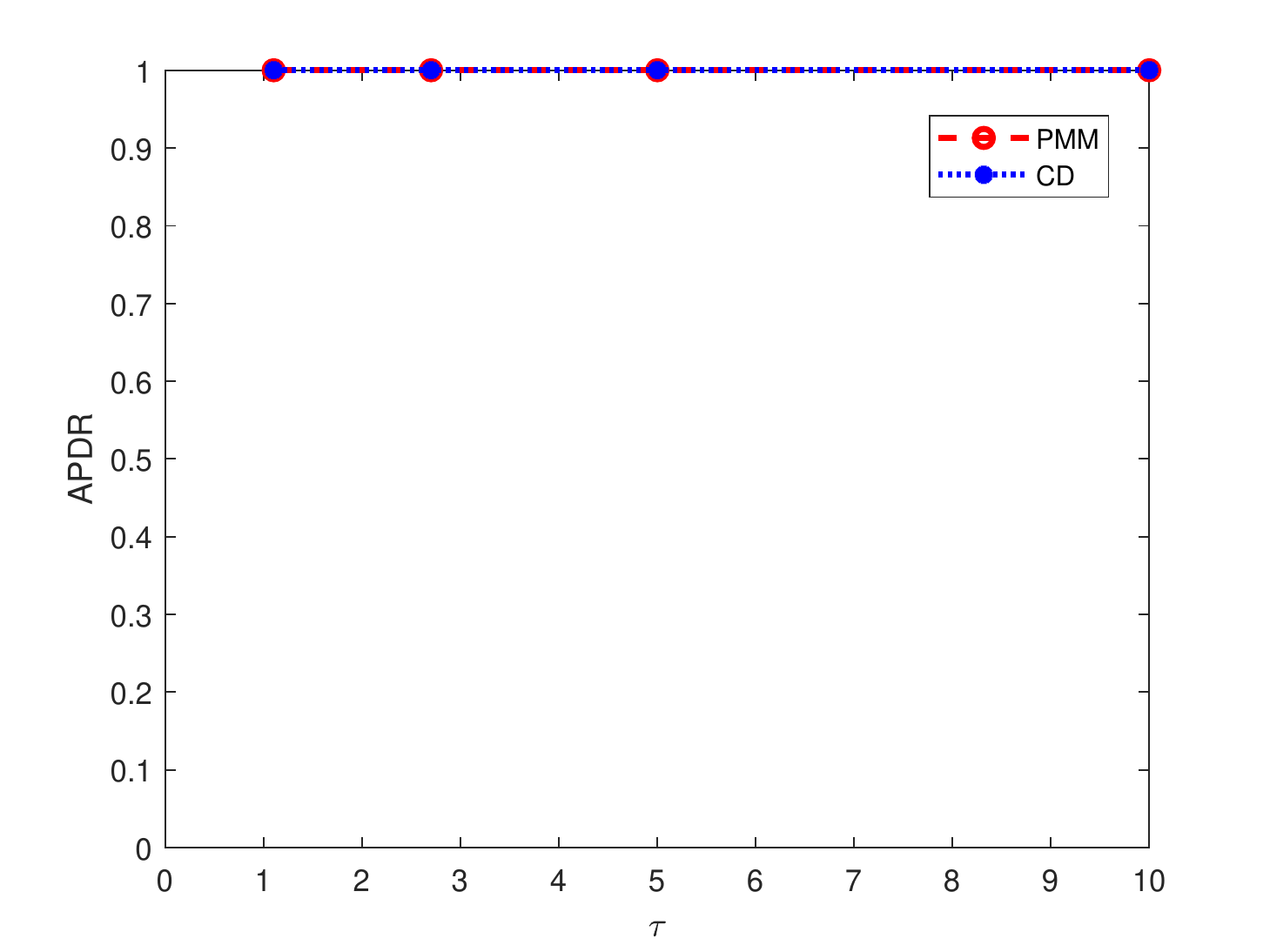}
\caption{{\small Numerical results of the influence of the model parameters on APDR}}
\label{fig:2}
\end{figure}

\begin{figure}[!ht]
\centering
\includegraphics[width=0.34\textwidth]{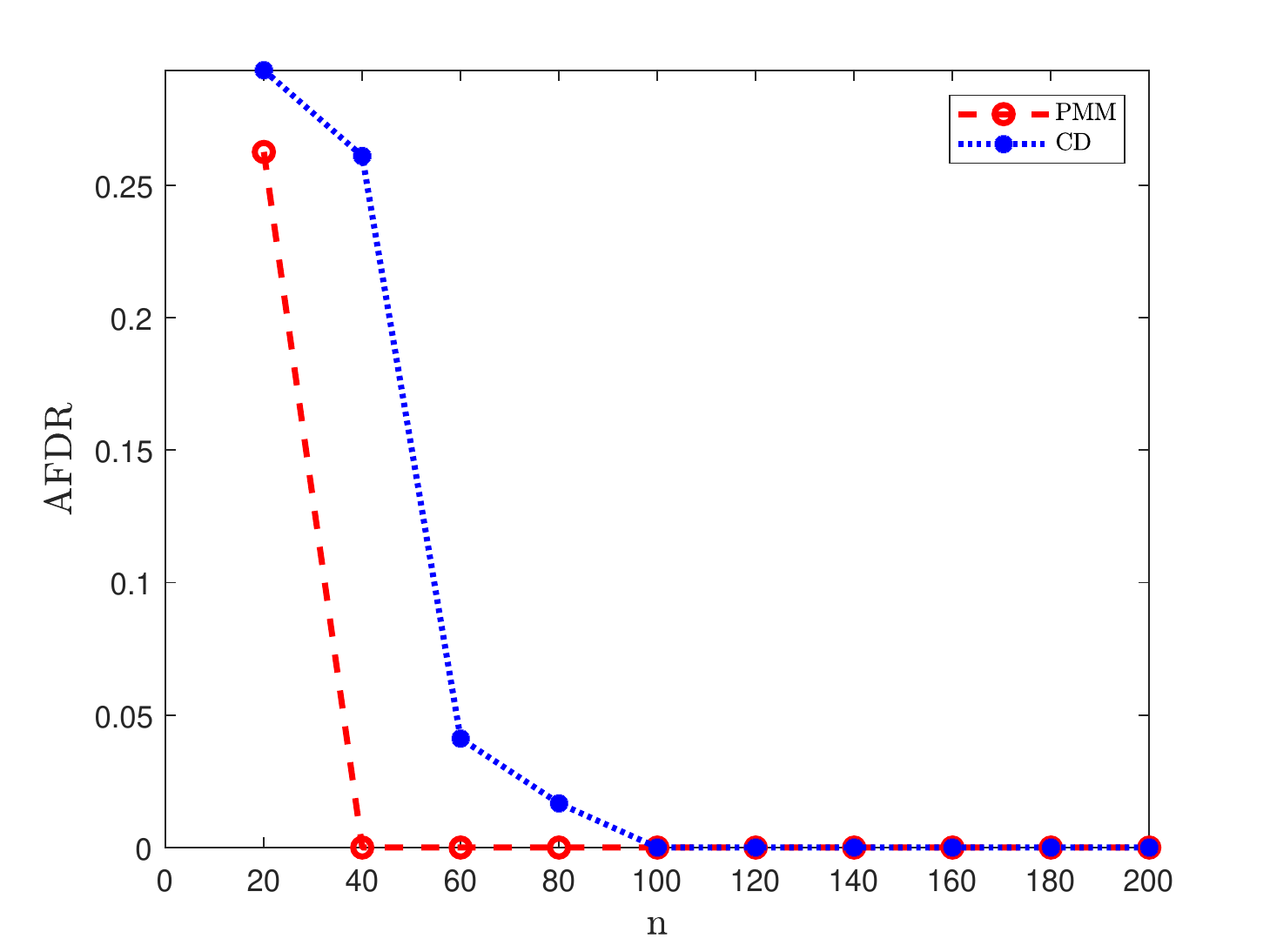}\hspace{-.4cm}
\includegraphics[width=0.34\textwidth]{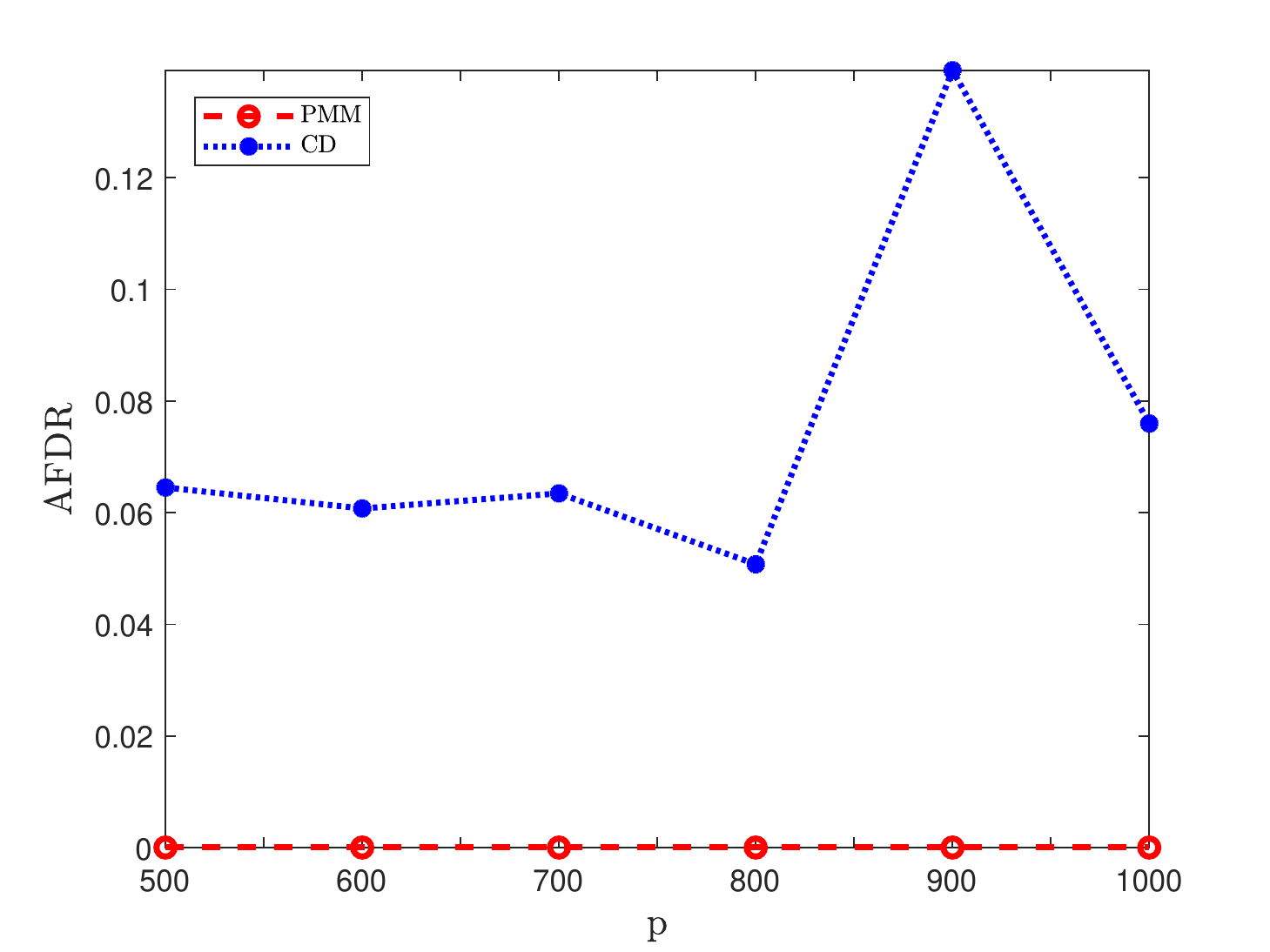}\hspace{-.4cm}
\includegraphics[width=0.34\textwidth]{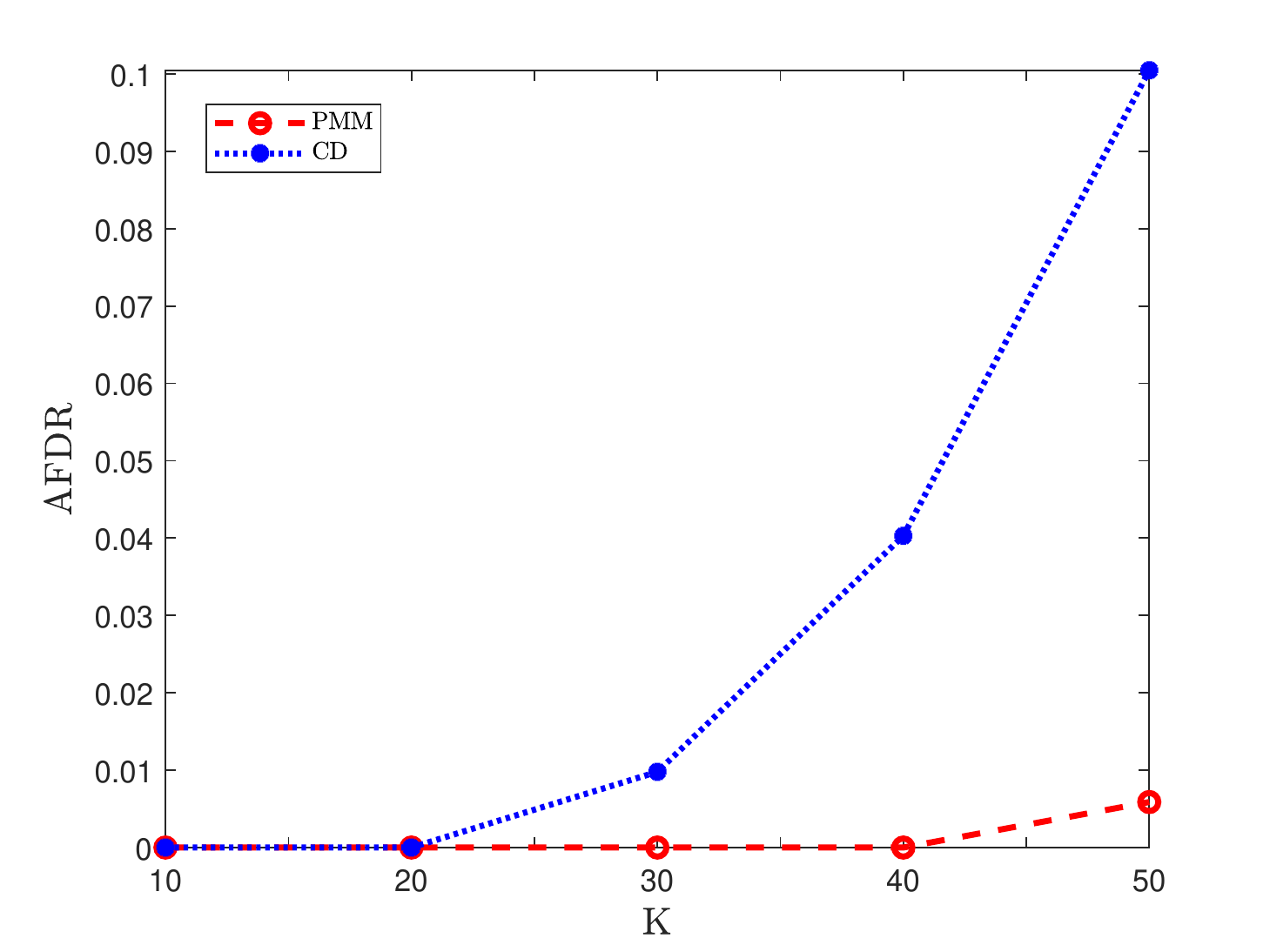}\\
\includegraphics[width=0.34\textwidth]{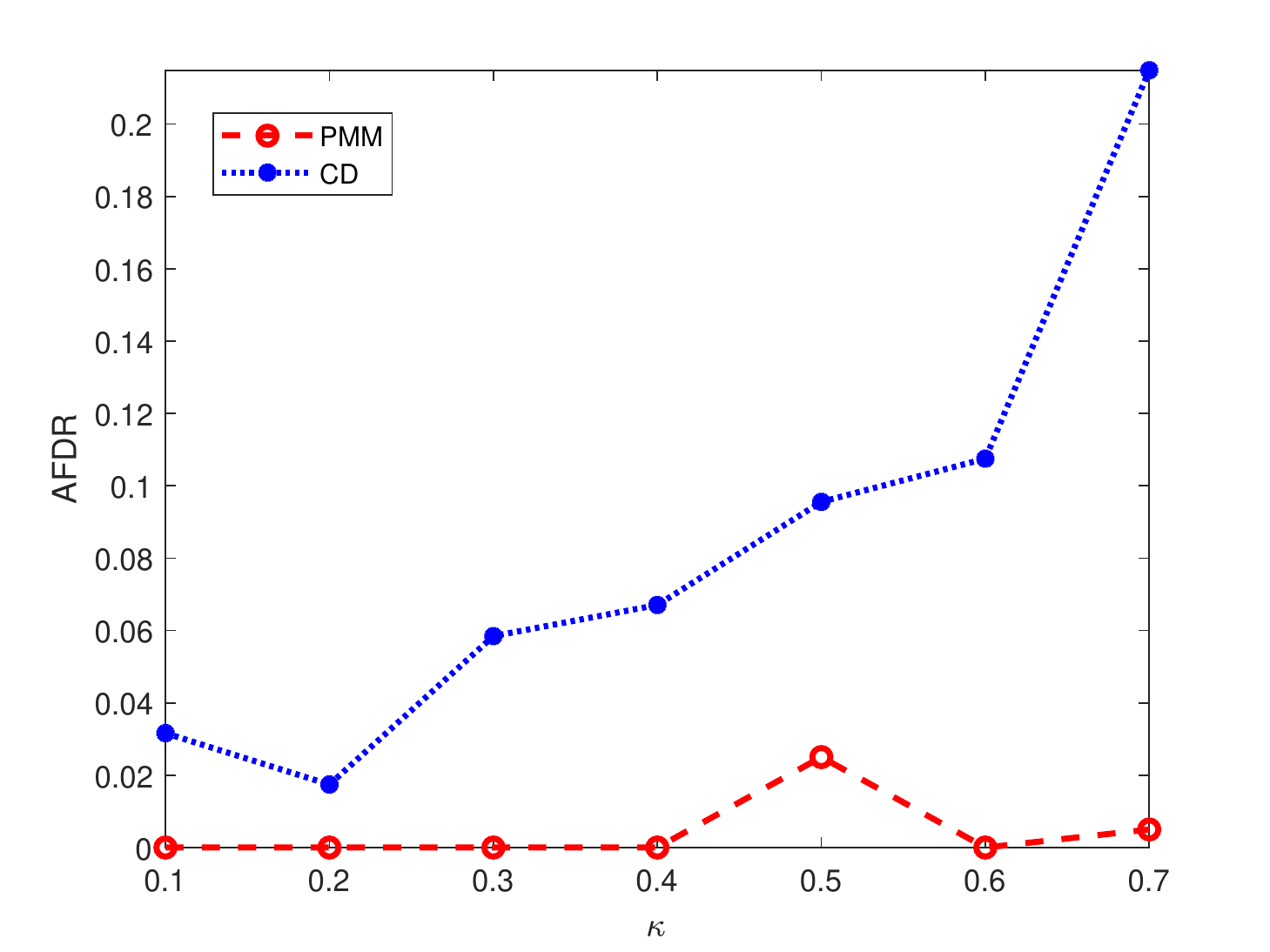}\hspace{-.4cm}
\includegraphics[width=0.34\textwidth]{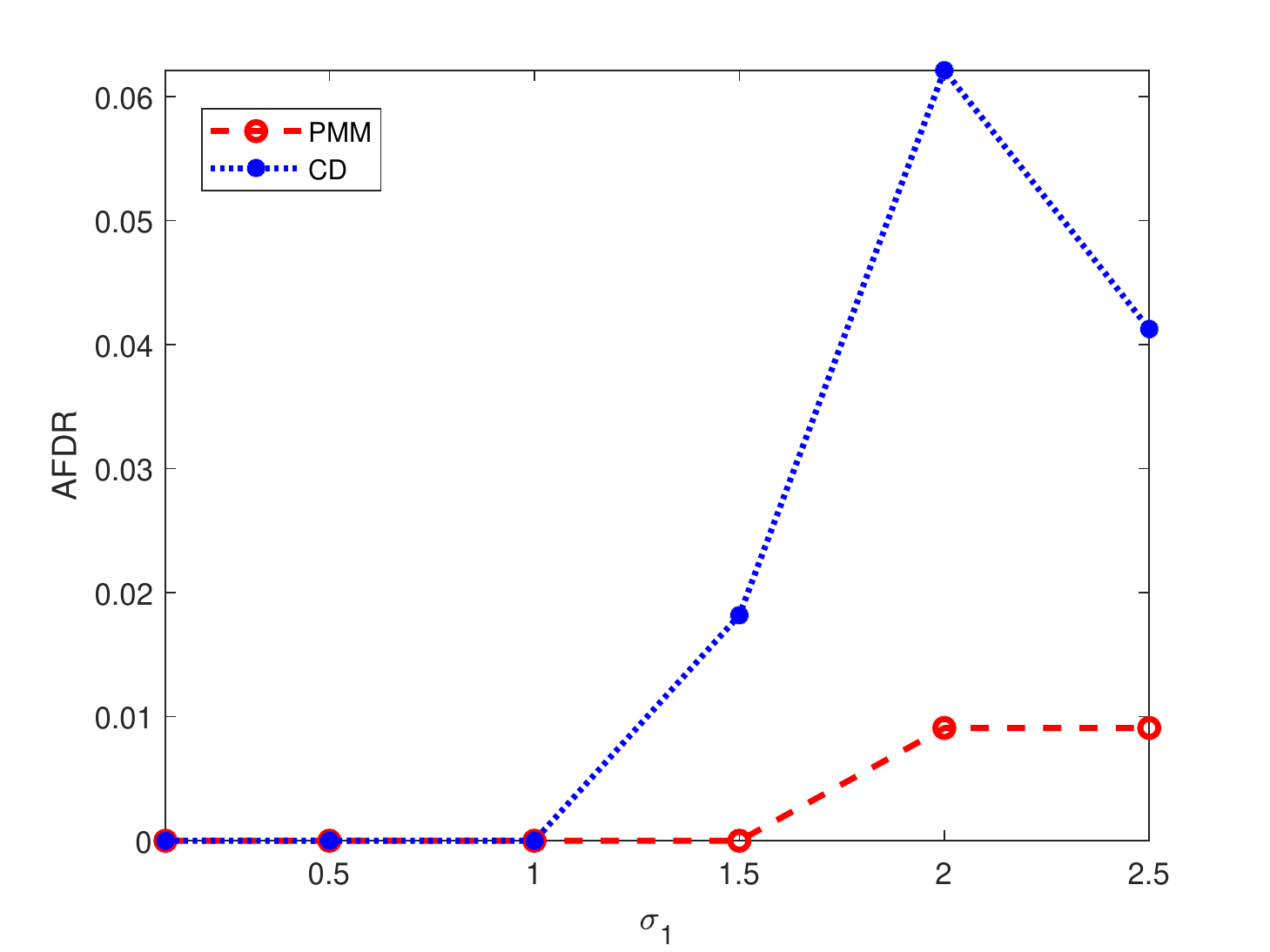}\hspace{-.4cm}
\includegraphics[width=0.34\textwidth]{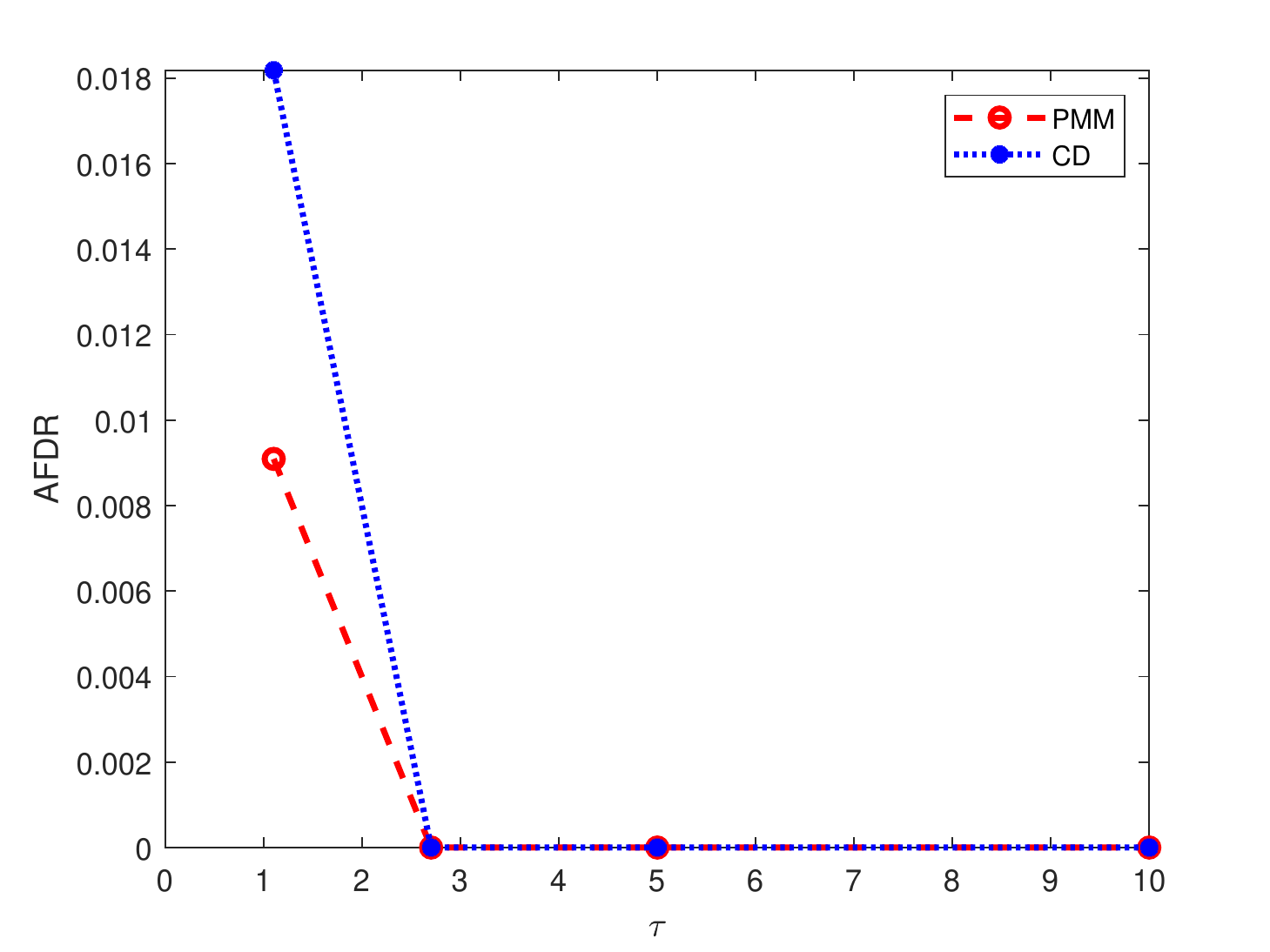}
\caption{{\small Numerical results of the influence of the model parameters on AFDR}}
\label{fig:3}
\end{figure}

\begin{figure}[!ht]
\centering
\includegraphics[width=0.34\textwidth]{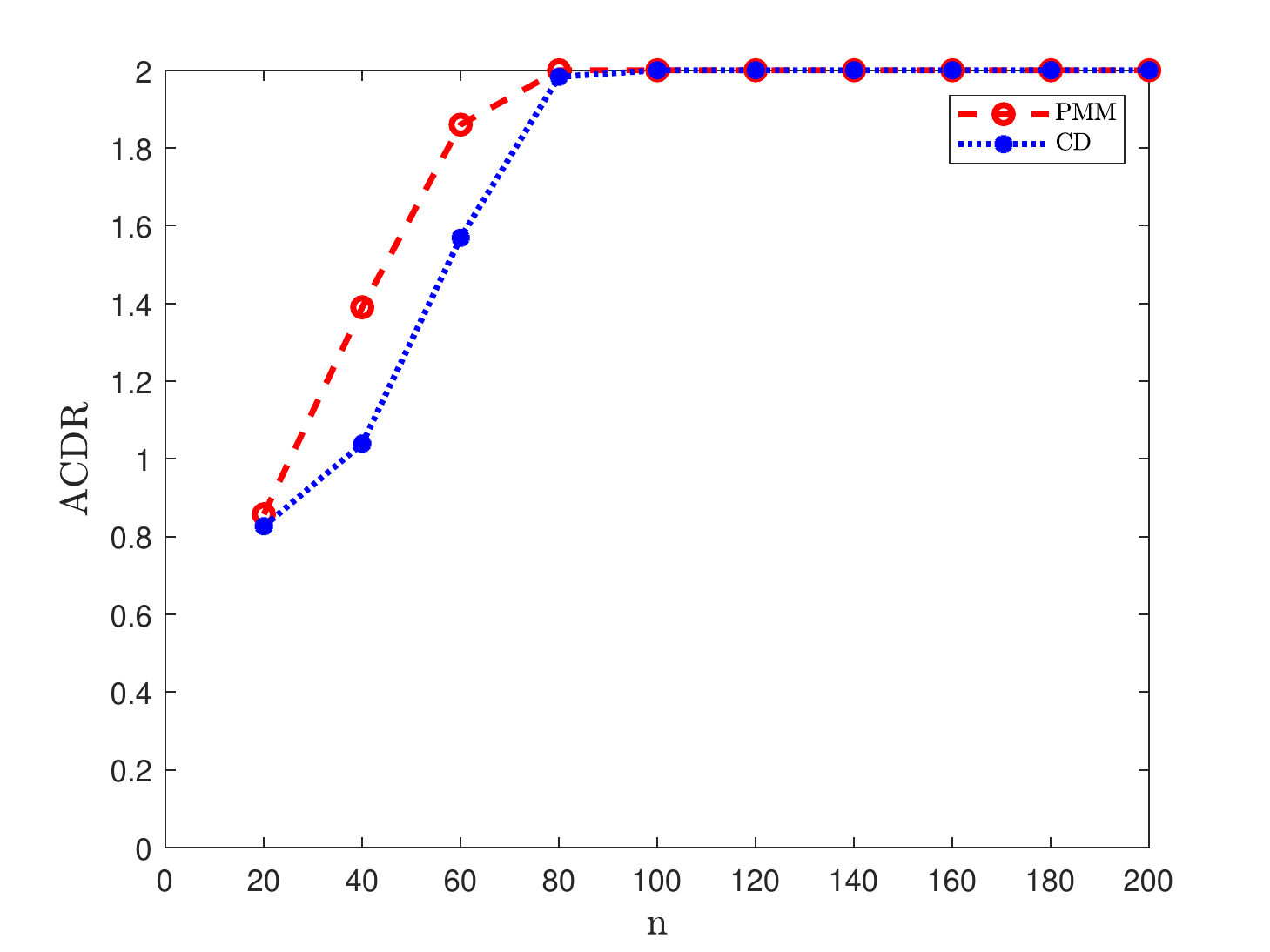}\hspace{-.4cm}
\includegraphics[width=0.34\textwidth]{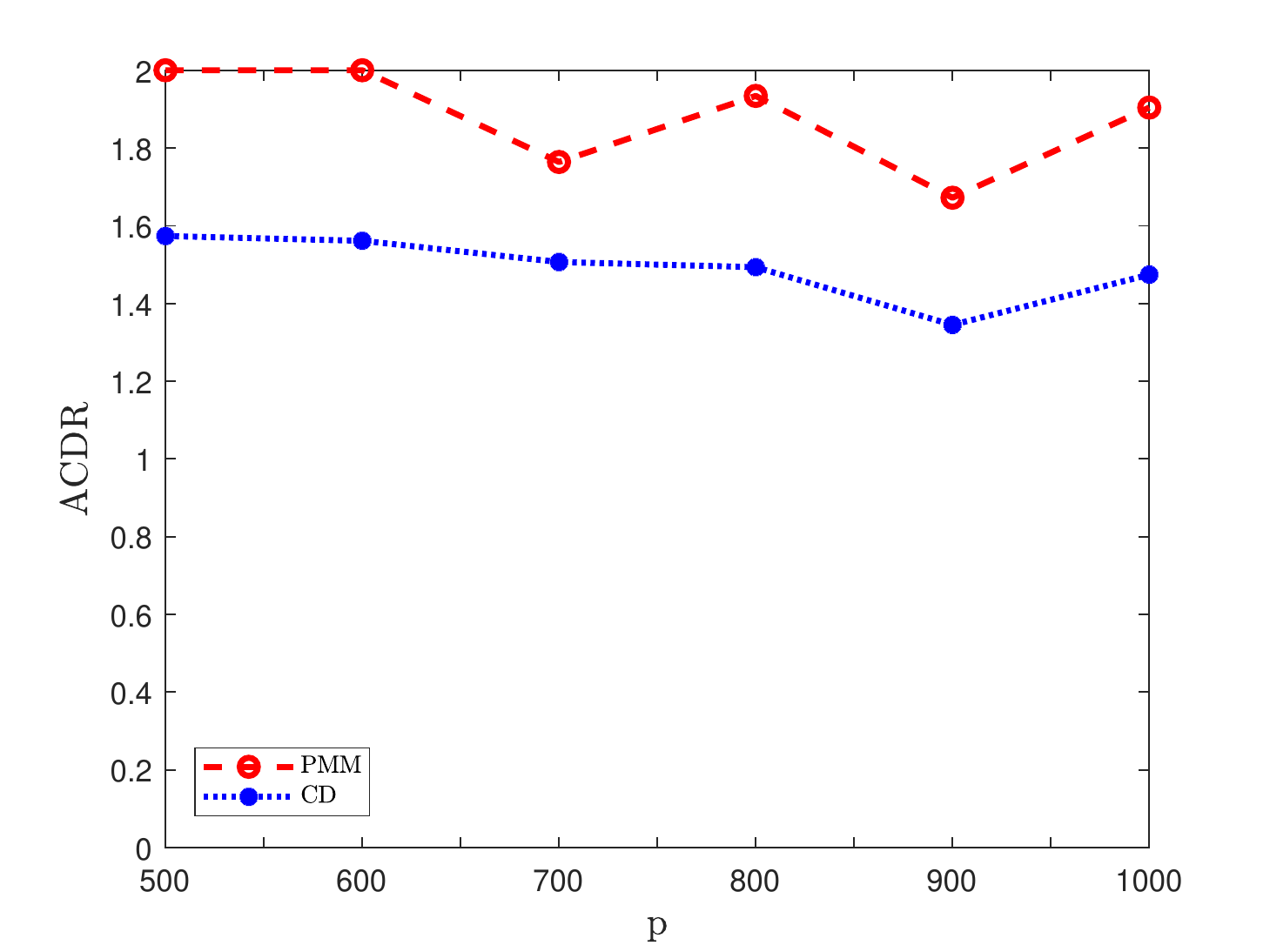}\hspace{-.4cm}
\includegraphics[width=0.34\textwidth]{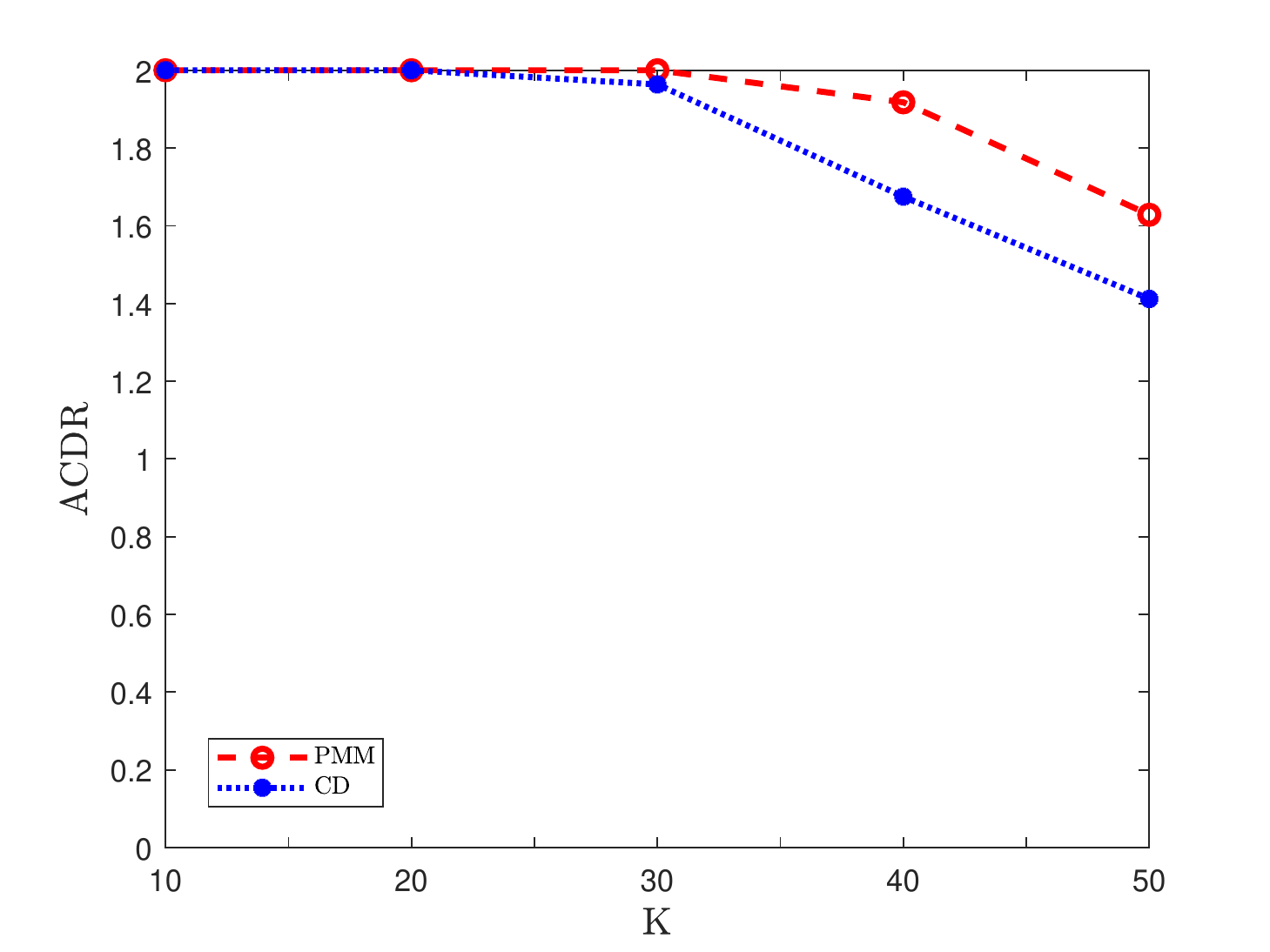}\\
\includegraphics[width=0.34\textwidth]{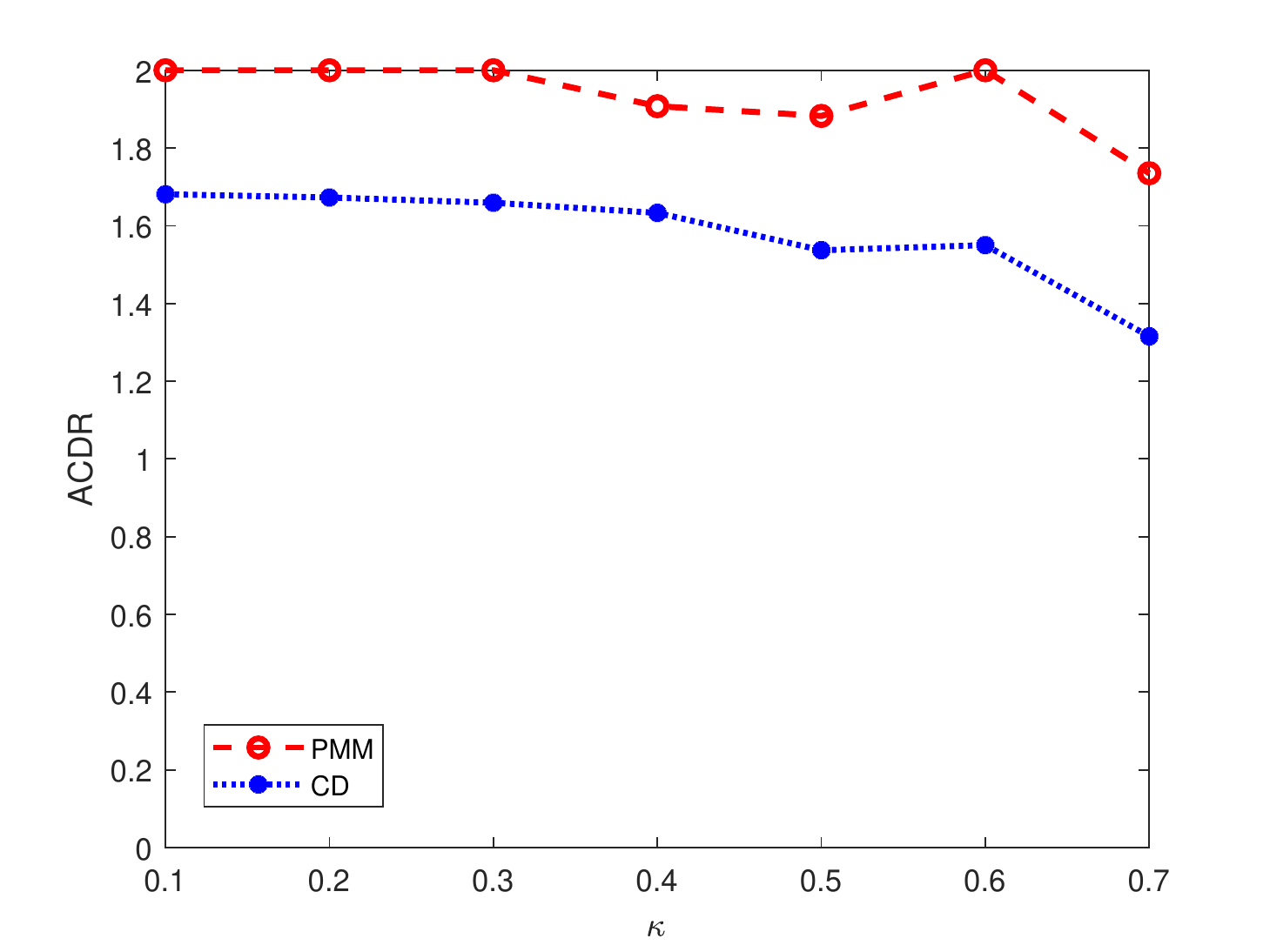}\hspace{-.4cm}
\includegraphics[width=0.34\textwidth]{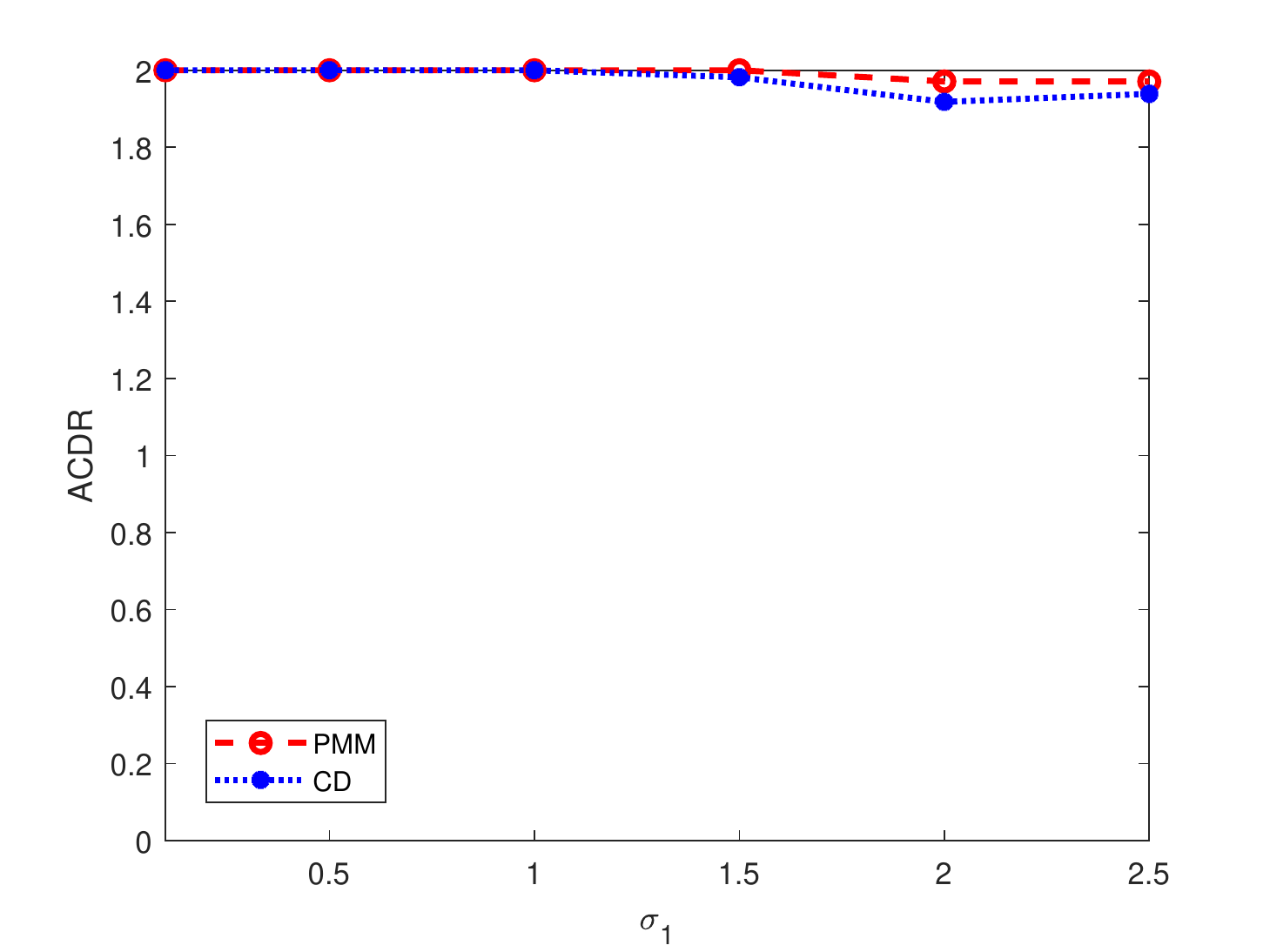}\hspace{-.4cm}
\includegraphics[width=0.34\textwidth]{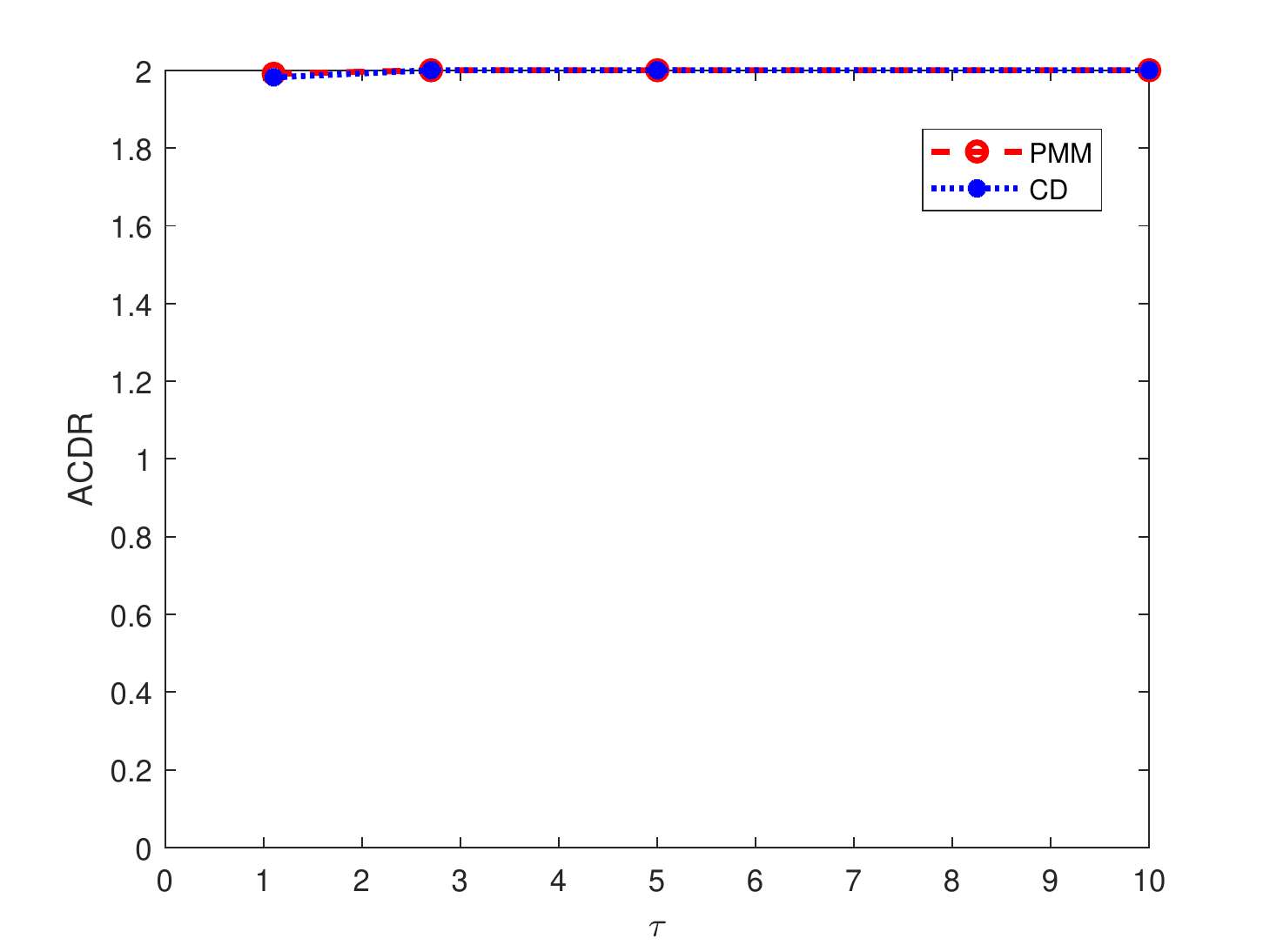}
\caption{{\small Numerical results of the influence of the model parameters on ACDR}}
\label{fig:4}
\end{figure}

From the definitions of APDR, AFDR and ACDR, we can conclude that the closer APDR approaches to 1, the closer AFDR approaches to 0, and the closer ACDR approaches to 2, the higher the solution quality. From Fig \ref{fig:2}-\ref{fig:4}, we can see that for the change interval of different parameters, PMM can always achieve more expected results. Therefore, compared with CD, PMM is more robust to the considered parameters for solving the MCP penalized least squares problems.
\subsubsection{Numerical comparison with real data}

In this subsection, we test CD and PMM algorithms with the test instances $(X,y)$ obtained from large-scale
regression problems in the LIBSVM data sets, which is available at \url{https://www.csie.ntu.edu.tw/~cjlin/libsvmtools/datasets}. These data sets are collected
from UCI, StatLib, Delve, 10-K Corpus, GWF01a. For computational efficiency, zero columns in $X$ are removed. As suggested in \cite{HJY2010}, in addition to the data sets \textbf{log1p.E2006.train} and \textbf{E2006.train}, we expand the original features of the remaining data sets by using polynomial basis functions over those features. For example, the last digit in \textbf{abalone7} indicates that an order 7 polynomial is used to generate the basis functions. This naming convention is also used in the rest of the expanded data sets. These test instances are quite difficult in terms of the problem dimensions and the largest eigenvalue of $X X^{\top}$, which is denoted as $\lambda_{max}(X X^{\top})$, one can refer to the first three columns of Table \ref{tab2}. It is worth noting that for these difficult real data, we appropriately reduce the accuracy requirements in the termination conditions. In addition to setting $R^2_{kkt}<5e-2$ for $\textbf{space\_ga9}$ and $R^2_{kkt}<8e-3$ for \textbf{bodyfat7}, we set $R^2_{kkt}<5e-3$ as the termination condition for all other data.

Table \ref{tab2} reports the detailed numerical results for CD and PMM in solving large-scale regression problems. In the table, ``NNZ" denotes the number of nonzeros in the estimated solution, and other symbols are the same as the previous simulation experiment. From the results in Table \ref{tab2}, we can see that PMM can solve all the instances to the desired accuracy despite the huge dimensions and the possibly badly conditioned data sets. More specifically, PMM is able to solve the
instance \textbf{log1p.E2006.train} with approximately 4.3 million features to accuracy $R^2_{kkt}=9.98e-6$ in 95 seconds. But CD only meets the accuracy requirement for \textbf{E2006.train}. In addition, for solving these data sets, CD needs much more time than PMM. For example, for the instance \textbf{cpusmall7}, we can see that PMM is at least 144 times faster than CD. The superior numerical performance of PMM indicates that it is a robust, high-performance solver for MCP/SCAD penalized high-dimensional linear regression problems.

\begin{table}[!ht]
\centering {\scriptsize\caption{Real results of CD and PMM algorithms}
\begin{tabular}{|c|c|c||c|c|c|c|c|}
\hline
Data name & n,p & $\lambda_{max}(X X^{\top})$ & Penalty & NNZ & Method & $R^2_{kkt}$ & Time        \\
\hline
\multirow{4}{*}{log1p.E2006.train} & \multirow{4}{*}{16087,4265669} & \multirow{4}{*}{5.86e+7} & \multirow{2}{*}{MCP} & \multirow{2}{*}{6} & CD & 3.91e-2 & 4.08e+3  \\
&  & & & & PMM  & 9.98e-6 & 9.39e+1 \\
&  & &  \multirow{2}{*}{SCAD} & \multirow{2}{*}{6}  & CD  & 5.36e-2 & 4.15e+3 \\
&  & &   & & PMM  & 9.97e-6 & 9.52e+1 \\
\cline{1-8}
\multirow{4}{*}{E2006.train} & \multirow{4}{*}{16087,150348} & \multirow{4}{*}{1.91e+5} & \multirow{2}{*}{MCP} & \multirow{2}{*}{6} & CD  & 2.50e-3 & 7.30e+1 \\
&  & & & &  PMM  & 4.90e-3 & 4.33e+1  \\
&  & &  \multirow{2}{*}{SCAD} & \multirow{2}{*}{6} & CD  & 2.50e-3 & 7.33e+1 \\
&  & & & & PMM  & 4.00e-3 & 3.68e+1 \\
\cline{1-8}
\multirow{4}{*}{abalone7} & \multirow{4}{*}{4177,6435} & \multirow{4}{*}{5.21e+5} & \multirow{2}{*}{MCP} & \multirow{2}{*}{7} & CD  & 9.66e-1 & 1.58e+1 \\
&  & & & & PMM  & 2.86e-4 & 1.51e+0 \\
&  & & \multirow{2}{*}{SCAD} & \multirow{2}{*}{7} & CD  & 9.64e-1 & 1.65e+1 \\
&  & & & & PMM & 2.77e-4 & 1.61e+0 \\
\cline{1-8}
\multirow{4}{*}{bodyfat7} &  \multirow{4}{*}{252,116280} & \multirow{4}{*}{5.29e+4} & \multirow{2}{*}{MCP} & \multirow{2}{*}{9} & CD  & 5.57e-2 & 1.18e+1  \\
&  & & & & PMM & 6.60e-3 & 2.30e+0 \\
&  & & \multirow{2}{*}{SCAD} & \multirow{2}{*}{1} & CD  & 7.63e-2 & 1.18e+1 \\
&  & & & & PMM  & 7.70e-3 & 2.23e+0 \\
\cline{1-8}
\multirow{4}{*}{cpusmall7} & \multirow{4}{*}{8192,50388} & \multirow{4}{*}{8.01e+7} & \multirow{2}{*}{MCP} & \multirow{2}{*}{1103} & CD  & 1.00e+0 & 2.97e+2 \\
&  & & & & PMM & 4.10e-3 & 2.05e+0 \\
&  & &  \multirow{2}{*}{SCAD} & \multirow{2}{*}{1105} & CD  & 1.00e+0 & 3.78e+2 \\
&  & & & & PMM  & 7.30e-4 & 4.57e+0  \\
\cline{1-8}
\multirow{4}{*}{housing7} & \multirow{4}{*}{506,77520} & \multirow{4}{*}{3.28e+5} & \multirow{2}{*}{MCP} & \multirow{2}{*}{44} & CD & 2.64e-2 & 1.48e+1  \\
&  & & & & PMM & 1.70e-3 & 2.38e-1 \\
&  & &  \multirow{2}{*}{SCAD} & \multirow{2}{*}{46} & CD & 3.62e-2 & 1.41e+1 \\
&  & &  & & PMM  & 8.74e-4 & 1.68e-1 \\
\cline{1-8}
\multirow{4}{*}{mg9} & \multirow{4}{*}{1385,5005} & \multirow{4}{*}{4.78e+3} & \multirow{2}{*}{MCP} & \multirow{2}{*}{9} & CD & 6.31e-1 & 2.29e+0  \\
&  & & & & PMM & 4.60e-3 & 1.56e+0  \\
&  & &  \multirow{2}{*}{SCAD} & \multirow{2}{*}{9} & CD  & 6.55e-1 & 2.21e+0  \\
&  & & & & PMM  & 4.50e-3 & 1.96e+0 \\
\cline{1-8}
\multirow{4}{*}{mpg7} & \multirow{4}{*}{392,3432} & \multirow{4}{*}{1.28e+4} & \multirow{2}{*}{MCP} & \multirow{2}{*}{26} & CD  & 4.92e-2 & 8.24e-1 \\
&  & & & & PMM  & 2.50e-3 & 3.65e-2 \\
&  & &  \multirow{2}{*}{SCAD} & \multirow{2}{*}{27} & CD & 6.75e-2 & 7.69e-1 \\
&  & & & & PMM  & 2.30e-3 & 3.75e-2 \\
\cline{1-8}
\multirow{4}{*}{pyrim5} & \multirow{4}{*}{74,169911} & \multirow{4}{*}{1.22e+6} & \multirow{2}{*}{MCP} & \multirow{2}{*}{326} & CD & 5.81e-2 & 1.07e+1  \\
&  & & & & PMM  & 4.92e-4 & 5.60e-2 \\
&  & &  \multirow{2}{*}{SCAD} & \multirow{2}{*}{327} & CD & 7.96e-2 & 9.94e+0 \\
&  & & & & PMM & 2.48e-4 & 5.37e-2 \\
\cline{1-8}
\multirow{4}{*}{$space\_ga9$} & \multirow{4}{*}{3107,5005} & \multirow{4}{*}{4.01e+3} & \multirow{2}{*}{MCP} & \multirow{2}{*}{8} & CD  & 1.89e-1 & 6.23e+0  \\
&  & & & & PMM  & 4.60e-2 & 1.99e+0  \\
&  & &  \multirow{2}{*}{SCAD} & \multirow{2}{*}{9} & CD  & 2.59e-1 & 6.26e+0 \\
&  & &  & & PMM  & 4.50e-2 & 2.24e+0 \\
\cline{1-8}
\multirow{4}{*}{triazines4} & \multirow{4}{*}{186,557845} & \multirow{4}{*}{2.08e+7} & \multirow{2}{*}{MCP} & \multirow{2}{*}{983} & CD  & 4.29e-2 & 9.61e+1  \\
&  & & & & PMM  & 7.39e-4 & 6.31e+1 \\
&  & &  \multirow{2}{*}{SCAD} & \multirow{2}{*}{983} & CD & 5.87e-2 & 9.48e+1 \\
&  & &  & & PMM & 1.86e-4 & 6.67e+1 \\
\hline
\end{tabular}\label{tab2}
}
\end{table}

%%%%%%%%%%%%%%%%%%%%%%%%%%%%%%%%%%%%%%%%%%%%%%%%%%%%%%%%%%%%%%%%%%%%%%%%%%%
\section{Conclusion}\label{con}
Based on the DC property of MCP and SCAD penalties, we developed a global two-stage algorithm for the MCP/SCAD penalized linear regression problems in high-dimensional
settings. A key idea for making the proposed algorithm to be efficient is to use the PDASC algorithm to solve the corresponding sub-problems. We established the global convergence of the proposed algorithm and verified the iterative sequence converges to a d-stationary point of the considered problems. Finally, a large number of inspiring numerical experiments have once again verified the effectiveness of the proposed algorithm.

Since each non-convex penalty can be expressed as the difference of two convex functions, the research in this paper can be directly extended to other non-convex penalized high-dimensional linear regression problems. In addition, extending the algorithm in this paper to the regression problems with other loss functions is also a very interesting and promising research direction.
%%%%%%%%%%%%%%%%%%%%%%%%%%%%%%%%%%%%%%%%%%%%%%%%%%%%%%%%%%%%%%%%%%%%
\section*{Acknowledgements}
The work of Zhou Yu is supported in part by the National Natural Science Foundation of China (Grant No. 11971170).

\bibliographystyle{Chicago}
\bibliography{pmm}

\begin{thebibliography}{}

\bibitem[\protect\citeauthoryear{Ahn, Pang, and Xin}{Ahn
  et~al.}{2017}]{APX2017}
Ahn, M., J.-S. Pang, and J.~Xin (2017).
\newblock {Difference-of-convex learning: directional stationarity, optimality,
  and sparsity}.
\newblock {\em SIAM Journal on Optimization\/}~{\em 27\/}(3), 1637--1665.

\bibitem[\protect\citeauthoryear{Boyd, Parikh, and Chu}{Boyd
  et~al.}{2011}]{BPC2011}
Boyd, S., N.~Parikh, and E.~Chu (2011).
\newblock {Distributed optimization and statistical learning via the
  alternating direction method of multipliers}.

\bibitem[\protect\citeauthoryear{Breheny and Huang}{Breheny and
  Huang}{2011}]{BH2011}
Breheny, P. and J.~Huang (2011).
\newblock {Coordinate descent algorithms for nonconvex penalized regression,
  with applications to biological feature selection}.
\newblock {\em The Annals of Applied Statistics\/}~{\em 5\/}(1), 232--253.

\bibitem[\protect\citeauthoryear{Candes and Tao}{Candes and Tao}{2005}]{CT2005}
Candes, E.~J. and T.~Tao (2005).
\newblock {Decoding by linear programming}.
\newblock {\em IEEE Transactions on Information Theory\/}~{\em 51\/}(12),
  4203--4215.

\bibitem[\protect\citeauthoryear{Chartrand}{Chartrand}{2007}]{C2007}
Chartrand, R. (2007).
\newblock {Exact reconstruction of sparse signals via nonconvex minimization}.
\newblock {\em IEEE Signal Processing Letters\/}~{\em 14\/}(10), 707--710.

\bibitem[\protect\citeauthoryear{Chen and Gu}{Chen and Gu}{2014}]{CG2014}
Chen, L. and Y.~Gu (2014).
\newblock {The convergence guarantees of a non-convex approach for sparse
  recovery}.
\newblock {\em IEEE Transactions on Signal Processing\/}~{\em 62\/}(15),
  3754--3767.

\bibitem[\protect\citeauthoryear{Chen, Donoho, and Saunders}{Chen
  et~al.}{}]{CDS2001}
Chen, S.~S., D.~L. Donoho, and M.~A. Saunders.
\newblock {Atomic decomposition by basis pursuit}.
\newblock {\em SIAM Review\/}~{\em 43\/}(1), 129--159.

\bibitem[\protect\citeauthoryear{Cui, Pang, and Sen}{Cui
  et~al.}{2018}]{CPS2018}
Cui, Y., J.-S. Pang, and B.~Sen (2018).
\newblock {Composite difference-max programs for modern statistical estimation
  problems}.
\newblock {\em SIAM Journal on Optimization\/}~{\em 28\/}(4), 3344--3374.

\bibitem[\protect\citeauthoryear{Donoho and Johnstone}{Donoho and
  Johnstone}{1995}]{DJ1995}
Donoho, D.~L. and I.~M. Johnstone (1995).
\newblock {Adapting to unknown smoothness via wavelet shrinkage}.
\newblock {\em Journal of the American Statistical Association\/}~{\em
  90\/}(432), 1200--1224.

\bibitem[\protect\citeauthoryear{Efron, Hastie, Johnstone, and
  Tibshirani}{Efron et~al.}{2004}]{EHJT2004}
Efron, B., T.~Hastie, I.~Johnstone, and R.~Tibshirani (2004).
\newblock {Least angle regression}.
\newblock {\em The Annals of Statistics\/}~{\em 32\/}(2), 407--499.

\bibitem[\protect\citeauthoryear{Fan and Li}{Fan and Li}{2001}]{FL2001}
Fan, J. and R.~Li (2001).
\newblock {Variable selection via nonconcave penalized likelihood and its
  oracle properties}.
\newblock {\em Journal of the American Statistical Association\/}~{\em
  96\/}(456), 1348--1360.

\bibitem[\protect\citeauthoryear{Fan, Jiao, and Lu}{Fan et~al.}{2014}]{FJL2014}
Fan, Q., Y.~Jiao, and X.~Lu (2014).
\newblock {A primal dual active set algorithm with continuation for compressed
  sensing}.
\newblock {\em IEEE Transactions on Signal Processing\/}~{\em 62\/}(23),
  6276--6285.

\bibitem[\protect\citeauthoryear{Frank and Friedman}{Frank and
  Friedman}{1993}]{FF1993}
Frank, L.~E. and J.~H. Friedman (1993).
\newblock {A statistical view of some chemometrics regression tools}.
\newblock {\em Technometrics\/}~{\em 35\/}(2), 109--135.

\bibitem[\protect\citeauthoryear{Fu}{Fu}{1998}]{F1998}
Fu, W.~J. (1998).
\newblock {Penalized regressions: the bridge versus the lasso}.
\newblock {\em Journal of Computational and Graphical Statistics\/}~{\em
  7\/}(3), 397--416.

\bibitem[\protect\citeauthoryear{Gong, Zhang, Lu, Huang, and Ye}{Gong
  et~al.}{2013}]{GZL2013}
Gong, P., C.~Zhang, Z.~Lu, J.~Huang, and J.~Ye (2013).
\newblock {A general iterative shrinkage and thresholding algorithm for
  non-convex regularized optimization problems}.
\newblock {\em Proceedings of the 30th International Conference on Machine
  Learning\/}~{\em 28\/}(2), 37--45.

\bibitem[\protect\citeauthoryear{Hinterm{\"u}ller, Ito, and
  Kunisch}{Hinterm{\"u}ller et~al.}{2002}]{HIK2002}
Hinterm{\"u}ller, M., K.~Ito, and K.~Kunisch (2002).
\newblock {The primal-dual active set strategy as a semismooth Newton method}.
\newblock {\em SIAM Journal on Optimization\/}~{\em 13\/}(3), 865--888.

\bibitem[\protect\citeauthoryear{Hofmann and Hohage}{Hofmann and
  Hohage}{2011}]{F2011}
Hofmann, B. and T.~Hohage (2011).
\newblock {Generalized Tikhonov regularization: Basic theory and comprehensive
  results on convergence rates}.
\newblock {\em Fakultat fur Mathematik\/}.

\bibitem[\protect\citeauthoryear{Huang, Jiao, Jin, Liu, Lu, and Yang}{Huang
  et~al.}{2021}]{HJJ2019}
Huang, J., Y.~Jiao, B.~Jin, J.~Liu, X.~Lu, and C.~Yang (2021).
\newblock {A unified primal dual active set algorithm for nonconvex sparse
  recovery}.
\newblock {\em Statistical Science\/}~{\em 36\/}(2), 215--238.

\bibitem[\protect\citeauthoryear{Huang, Jia, Yu, Chun, Maniatis, and
  Naik}{Huang et~al.}{2010}]{HJY2010}
Huang, L., J.~Jia, B.~Yu, B.-G. Chun, P.~Maniatis, and M.~Naik (2010).
\newblock {Predicting execution time of computer programs using sparse
  polynomial regression}.
\newblock {\em Advances in Neural Information Processing Systems\/}~{\em 23},
  883--891.

\bibitem[\protect\citeauthoryear{Hunter and Li}{Hunter and Li}{2005}]{HL2005}
Hunter, D.~R. and R.~Li (2005).
\newblock {Variable selection using MM algorithms}.
\newblock {\em The Annals of Statistics\/}~{\em 33\/}(4), 1617--1642.

\bibitem[\protect\citeauthoryear{Le~Thi, Dinh, Le, and Vo}{Le~Thi
  et~al.}{2015}]{LPLV2015}
Le~Thi, H.~A., T.~P. Dinh, H.~M. Le, and X.~T. Vo (2015).
\newblock {DC approximation approaches for sparse optimization}.
\newblock {\em European Journal of Operational Research\/}~{\em 244\/}(1),
  26--46.

\bibitem[\protect\citeauthoryear{Lee, Sun, and Saunders}{Lee
  et~al.}{2014}]{LSS2014}
Lee, J.~D., Y.~Sun, and M.~A. Saunders (2014).
\newblock {Proximal Newton-type methods for minimizing composite functions}.
\newblock {\em SIAM Journal on Optimization\/}~{\em 24\/}(3), 1420--1443.

\bibitem[\protect\citeauthoryear{Li and Pong}{Li and Pong}{2018}]{LP2017}
Li, G. and T.~K. Pong (2018).
\newblock {Calculus of the exponent of Kurdyka-Lojasiewicz inequality and its
  applications to linear convergence of first-order methods}.
\newblock {\em Foundations of Computational Mathematics\/}~{\em 18\/}(5),
  1199--1232.

\bibitem[\protect\citeauthoryear{Li, Yang, Ge, Haupt, Zhang, and Zhao}{Li
  et~al.}{2017}]{LYG2017}
Li, X., L.~Yang, J.~Ge, J.~Haupt, T.~Zhang, and T.~Zhao (2017).
\newblock {On quadratic convergence of DC proximal Newton algorithm in
  nonconvex sparse learning}.
\newblock {\em Advances in Neural Information Processing Systems\/}~{\em 30},
  2742--2752.

\bibitem[\protect\citeauthoryear{Li, Sun, and Toh}{Li et~al.}{2018}]{LST2018}
Li, X.~D., D.~F. Sun, and K.~C. Toh (2018).
\newblock {A Highly Efficient Semismooth Newton Augmented Lagrangian Method for
  Solving LASSO Problems}.
\newblock {\em SIAM Journal on Optimization\/}~{\em 28\/}(1), 433--458.

\bibitem[\protect\citeauthoryear{Luo and Chen}{Luo and Chen}{2014}]{LC2014}
Luo, S. and Z.~Chen (2014).
\newblock {Sequential Lasso cum EBIC for feature selection with ultra-high
  dimensional feature space}.
\newblock {\em Journal of the American Statistical Association\/}~{\em
  109\/}(507), 1229--1240.

\bibitem[\protect\citeauthoryear{Mazumder, Friedman, and Hastie}{Mazumder
  et~al.}{2011}]{MFH2011}
Mazumder, R., J.~H. Friedman, and T.~Hastie (2011).
\newblock {Sparsenet: Coordinate descent with nonconvex penalties}.
\newblock {\em Journal of the American Statistical Association\/}~{\em
  106\/}(495), 1125--1138.

\bibitem[\protect\citeauthoryear{Meinshausen and Buhlmann}{Meinshausen and
  Buhlmann}{2006}]{MB2006}
Meinshausen, N. and P.~Buhlmann (2006).
\newblock {High-dimensional graphs and variable selection with the lasso}.
\newblock {\em The Annals of Statistics\/}~{\em 34\/}(3), 1436--1462.

\bibitem[\protect\citeauthoryear{Micchelli, Shen, and Xu}{Micchelli
  et~al.}{2011}]{MSX2011}
Micchelli, C.~A., L.~Shen, and Y.~Xu (2011).
\newblock {Proximity algorithms for image models: denoising}.
\newblock {\em Inverse Problems\/}~{\em 27\/}(4), 045009.

\bibitem[\protect\citeauthoryear{Natarajan}{Natarajan}{1995}]{N1995}
Natarajan, B.~K. (1995).
\newblock {Sparse approximate solutions to linear systems}.
\newblock {\em SIAM Journal on Computing\/}~{\em 24\/}(2), 227--234.

\bibitem[\protect\citeauthoryear{Pang, Razaviyayn, and Alvarado}{Pang
  et~al.}{2017}]{PRA2016}
Pang, J.-S., M.~Razaviyayn, and A.~Alvarado (2017).
\newblock {Computing B-stationary points of nonsmooth DC programs}.
\newblock {\em Mathematics of Operations Research\/}~{\em 42\/}(1), 95--118.

\bibitem[\protect\citeauthoryear{Rockafellar}{Rockafellar}{2015}]{R1996}
Rockafellar, R.~T. (2015).
\newblock {Convex analysis}.

\bibitem[\protect\citeauthoryear{Shi, Huang, Jiao, and Yang}{Shi
  et~al.}{2018}]{SHJ2018}
Shi, Y., J.~Huang, Y.~Jiao, and Q.~Yang (2018).
\newblock {Semi-smooth Newton algorithm for non-convex penalized linear
  regression}.
\newblock {\em arXiv preprint arXiv:1802.08895\/}.

\bibitem[\protect\citeauthoryear{Tang, Wang, Sun, and Toh}{Tang
  et~al.}{2020}]{TWST2020}
Tang, P., C.~Wang, D.~Sun, and K.-C. Toh (2020).
\newblock {A sparse semismooth Newton based proximal majorization-minimization
  algorithm for nonconvex square-root-loss regression problems}.
\newblock {\em Journal of Machine Learning Research\/}~{\em 21\/}(226), 1--38.

\bibitem[\protect\citeauthoryear{Tibshirani}{Tibshirani}{1996}]{T1996}
Tibshirani, R. (1996).
\newblock {Regression shrinkage and selection via the lasso}.
\newblock {\em Journal of the Royal Statistical Society: Series B
  (Methodological)\/}~{\em 58\/}(1), 267--288.

\bibitem[\protect\citeauthoryear{Wang, Kim, and Li}{Wang
  et~al.}{2013}]{WKL2013}
Wang, L., Y.~Kim, and R.~Li (2013).
\newblock {Calibrating non-convex penalized regression in ultra-high
  dimension}.
\newblock {\em The Annals of Statistics\/}~{\em 41\/}(5), 2505--2536.

\bibitem[\protect\citeauthoryear{Wu and Lange}{Wu and Lange}{2008}]{WL2008}
Wu, T.~T. and K.~Lange (2008).
\newblock {Coordinate descent algorithms for lasso penalized regression}.
\newblock {\em The Annals of Applied Statistics\/}~{\em 2\/}(1), 224--244.

\bibitem[\protect\citeauthoryear{Zhang}{Zhang}{2010a}]{Z2010}
Zhang, C.-H. (2010a).
\newblock {Nearly unbiased variable selection under minimax concave penalty}.
\newblock {\em The Annals of Statistics\/}~{\em 38\/}(2), 894--942.

\bibitem[\protect\citeauthoryear{Zhang}{Zhang}{2010b}]{Z2010b}
Zhang, T. (2010b).
\newblock {Analysis of multi-stage convex relaxation for sparse
  regularization}.
\newblock {\em Journal of Machine Learning Research\/}~{\em 11\/}(3),
  1081--1107.

\bibitem[\protect\citeauthoryear{Zhao and Yu}{Zhao and Yu}{2006}]{ZY2006}
Zhao, P. and B.~Yu (2006).
\newblock {On model selection consistency of Lasso}.
\newblock {\em Journal of Machine Learning Research\/}~{\em 7}, 2541--2563.

\bibitem[\protect\citeauthoryear{Zou and Li}{Zou and Li}{2008}]{ZL2008}
Zou, H. and R.~Li (2008).
\newblock {One-step sparse estimates in nonconcave penalized likelihood
  models}.
\newblock {\em The Annals of Statistics\/}~{\em 36\/}(4), 1509--1533.

\end{thebibliography}

\end{document}